\documentclass[a4paper,reqno]{amsart}
\usepackage[all]{xy}           %For commutative diagrams
\usepackage{amssymb}           %For double arrows
\usepackage{hyperref}
\usepackage{eucal}
\usepackage{graphicx}
\usepackage{epsfig}
\usepackage[usenames,dvipsnames]{color}
\numberwithin{equation}{section}

\newtheorem{definition}{Definition}[section]
\newtheorem{theorem}[definition]{Theorem}
\newtheorem{proposition}[definition]{Proposition}
\newtheorem{corollary}[definition]{Corollary}
\newtheorem{remarkth}[definition]{Remark}

\newenvironment{remark}{\begin{remarkth}\upshape}{\hfill$\diamond$\end{remarkth}}
\renewcommand{\emph}[1]{{\bfseries\itshape{#1}}}
\newcommand{\R}{\mathbb{R}}      %Numeros reales
\newcommand{\T}{\mathbb{T}}

\newcount\ancho \newcount\anchom \newcount\anchoa
\newcount\anchob \newcount\altura

\newcommand{\ltilde}[3][0]{\altura=0 \advance\altura by #1
           \ancho=#2 \anchom=\ancho \divide\anchom by 2
           \anchoa=\ancho \divide\anchoa by 4
           \anchob=\anchom \advance\anchob by \anchoa
           \kern-3pt \begin{array}[b]{c}
           \begin{picture}(1,1)(\anchom,-\altura)
        \qbezier(0,2)(\anchoa,5)(\anchom,2)
        \qbezier(\anchom,2)(\anchob,-1)(\ancho,4)
        \qbezier(0,2)(\anchoa,4.5)(\anchom,1.8)
        \qbezier(\anchom,1.8)(\anchob,-1.5)(\ancho,4)
       \end{picture} \\[-4pt]{#3}
                       \end{array} \kern-4pt    }

\newcommand{\lhat}[3][0]{\altura=0 \advance\altura by #1
           \ancho=#2 \anchom=\ancho \divide\anchom by 2
           \anchoa=\ancho \divide\anchoa by 4
           \anchob=\anchom \advance\anchob by \anchoa
           \kern-3pt \begin{array}[b]{c}
           \begin{picture}(1,1)(\anchom,-\altura)
        \qbezier(0,2)(\anchoa,4)(\anchom,6)
        \qbezier(\anchom,6)(\anchob,4)(\ancho,2)
        \qbezier(0,2)(\anchoa,3.8)(\anchom,5.6)
        \qbezier(\anchom,5.6)(\anchob,3.8)(\ancho,2)
       \end{picture} \\[-4pt] {#3}
                       \end{array} \kern-4pt    }

\newcommand{\lcf}{\lbrack\! \lbrack}
\newcommand{\rcf}{\rbrack\! \rbrack}

\newcommand{\lvec}[1]{\overleftarrow{#1}}

\newcommand{\qquand}{\qquad\text{and}\qquad}

\renewcommand{\sec}[1]{{\rm Sec}(#1)}

\makeatletter
\newcommand\prol{\@ifstar{\@proldf}{\@prolpf}}  %% if * dual else primal
\def\@prolpf{\@ifnextchar[{\@prolpf@wrt}{\@prolpf@}}
\def\@prolpf@wrt[#1]#2{\@ifnextchar[{\@prolpf@wrt@at{#1}{#2}}{\@prolpf@wrt@{#1}{#2}}}
\def\@prolpf@wrt@at#1#2[#3]{\prolsymbol^{#1}_{#3}#2}
\def\@prolpf@wrt@#1#2{\prolsymbol^{#1}#2}
\def\@prolpf@#1{\@ifnextchar[{\@prolpf@at{#1}}{\@prolpf@@{#1}}}
\def\@prolpf@at#1[#2]{\prolsymbol_{#2}#1}
\def\@prolpf@@#1{\prolsymbol#1}
\def\@proldf{\@ifnextchar[{\@proldf@wrt}{\@proldf@}}
\def\@proldf@wrt[#1]#2{\@ifnextchar[{\@proldf@wrt@at{#1}{#2}}{\@proldf@wrt@{#1}{#2}}}
\def\@proldf@wrt@at#1#2[#3]{\prolsymbol^{*#1}_{#3}#2}
\def\@proldf@wrt@#1#2{\prolsymbol^{*#1}#2}
\def\@proldf@#1{\@ifnextchar[{\@proldf@at{#1}}{\@proldf@@{#1}}}
\def\@proldf@at#1[#2]{\prolsymbol^*_{#2}#1}
\def\@proldf@@#1{\prolsymbol^*#1}
\def\prolsymbol{\mathcal{T}}
\makeatother

\newcommand{\cnabla}{\check{\nabla}}

\def\lcf{\lbrack\! \lbrack}
\def\rcf{\rbrack\! \rbrack}
\def\cnabla{\check{\nabla}}

\let\sec\Sec

\begin{document}
{\Large

\title[Hamilton-Jacobi Equation and nonholonomic Mechanics]{Linear
almost Poisson structures and Hamilton-Jacobi Equation. Applications
to nonholonomic Mechanics}

\author[M.\ de Le\'on]{Manuel de Le\'on}
\address{M.\ de Le\'on:
Instituto de Ciencias Matem\'aticas (CSIC-UAM-UC3M-UCM),  Serrano 123, 28006
Madrid, Spain} \email{mdeleon@imaff.cfmac.csic.es}

\author[J.\ C.\ Marrero]{Juan C.\ Marrero}
\address{Juan C.\ Marrero:
ULL-CSIC Geometr\'{\i}a Diferencial y Mec\'anica Geom\'etrica\\
Departamento de Matem\'atica Fundamental, Facultad de
Ma\-te\-m\'a\-ti\-cas, Universidad de la Laguna, La Laguna,
Tenerife, Canary Islands, Spain} \email{jcmarrer@ull.es}

\author[D.\ Mart\'{\i}n de Diego]{David Mart\'{\i}n de Diego}
\address{D.\ Mart\'{\i}n de Diego:
Instituto de Ciencias Matem\'aticas (CSIC-UAM-UC3M-UCM), Serrano 123, 28006
Madrid, Spain} \email{d.martin@imaff.cfmac.csic.es}

\thanks{This work has been partially supported by MEC (Spain)
Grants MTM 2006-03322, MTM 2007-62478, project "Ingenio
Mathematica" (i-MATH) No. CSD 2006-00032 (Consolider-Ingenio 2010)
and S-0505/ESP/0158 of the Comunidad de Madrid. We gratefully
acknowledge helpful comments and suggestions of Anthony Bloch,
Eduardo Mart\'{\i}nez and Tomoki Oshawa. The authors also thank
the referees, who suggested important improvements upon the first
versions of our paper}



\begin{abstract}
In this paper, we  study the underlying geometry in the classical
Hamilton-Jacobi equation. The proposed formalism is also valid for
nonholonomic systems. We first introduce the essential geometric
ingredients: a vector bundle,  a linear almost Poisson structure
and a Hamiltonian function, both on the dual bundle (a Hamiltonian
system). From them, it is possible to formulate the
Hamilton-Jacobi equation, obtaining as a particular case, the
classical theory. The main application in this paper is to
nonholonomic mechanical systems. For it, we first construct the
linear almost Poisson structure on the dual space of the vector
bundle of admissible directions, and then, apply the
Hamilton-Jacobi theorem. Another important fact in our paper  is the use of the orbit theorem
to symplify the Hamilton-Jacobi equation, the introduction of the notion of morphisms preserving the
Hamiltonian system; indeed, this concept will be very useful to
treat with reduction procedures for systems with symmetries.
Several detailed examples are given to illustrate the utility of these new developments.
\end{abstract}

\maketitle

\tableofcontents

\section{Introduction}

The standard Hamilton-Jacobi equation is the first-order, non-linear partial differential equation,
 \begin{equation}\label{hj1}
\frac{\partial S}{\partial t} + H(q^A, \frac{\partial S}{\partial
q^A}) = 0,
\end{equation}
 for a function $S(t, q^A)$ (called the principal function) and
 where $H$ is the Hamiltonian function of the system. Taking $S(t,
 q^A) = W(q^A) - t E$, where $E$ is a constant, we rewrite the previous equations as
\begin{equation}\label{hj2}
H(q^A, \frac{\partial W}{\partial q^A}) = E,
\end{equation}
where $W$ is called the characteristic function. Equations
(\ref{hj1}) and (\ref{hj2}) are indistinctly referred as the
Hamilton-Jacobi equation (see \cite{AbMa,Godb}; see also
\cite{CaGrMaMaMuRo} for a recent geometrical approach).

The motivation of the present paper is to extend this theory for
the case of nonholonomic mechanical systems, that is, those
mechanical systems subject to linear constraints on the
velocities. In Remark \ref{history} of our paper, we carefully
summarize previous approaches to this subject. These tried  to
adapt the standard  Hamilton-Jacobi equations for systems without
constraints to the nonholonomic setting. But for nonholonomic
mechanics is necessary to take into account that the dynamics is
obtained from  an almost Poisson bracket, that is, a bracket not
satisfying the Jacobi identity.  In this direction, in a recent
paper \cite{IgLeMa}, the authors have developed a new approach
which permits to extend the Hamilton-Jacobi equation to
nonholonomic mechanical systems. However, the expression of the
corresponding Hamilton-Jacobi equation is  far from the standard
Hamilton-Jacobi equation for unconstrained systems. This fact has
motivated the present discussion since it was necessary to
understand the underlying geometric structure in the proposed
Hamilton-Jacobi equation for nonholonomic systems.

To go further in this direction, we need a new framework, which
captures the non-Hamiltonian essence of a nonholonomic problem.
Thus, we have considered a more general minimal ``Hamiltonian"
scenario. The starting point is a vector bundle $\tau_D : D
\longrightarrow Q$ such that its dual vector bundle $\tau_{D^*} :
D^* \longrightarrow Q$ is equipped with a linear almost Poisson
bracket $\{\cdot , \cdot \}_{D^*}$, that is, a linear bracket
satisfying all the usual properties of a Poisson bracket except
the  Jacobi identity. The existence of such bracket is equivalent
to the existence of an skew-symmetric algebroid structure $(\lcf
\cdot, \cdot \rcf_D, \rho_D)$ on $\tau_D : D \longrightarrow Q$
(i.e. a Lie algebroid structure eliminating the integrability
property), or even, the existence of an almost differential $d^D$
on $\tau_D : D \longrightarrow Q$, that is, an operator $d^D$
which acts on the ``forms" on $D$ and it satisfies all the
properties of an standard differential except that $(d^D)^2$ is
not, in general, zero. We remark that skew-symmetric algebroid
structures are almost Lie structures in the terminology of
\cite{Po} (see also \cite{PoPo}) and that the one-to-one
correspondence between skew-symmetric algebroids and almost
differentials was obtained in \cite{Po}. We also note that an
skew-symmetric algebroid also is called a pre-Lie algebroid in the
terminology introduced in some papers (see, for instance,
\cite{GrUr0,GrUr,KoMa}) and the relation between linear almost
Poisson brackets and skew-symmetric algebroid structures was
discussed in these papers (see also \cite{GrGr,GrUrGr} for some
applications to Classical Mechanics).

In this framework, a Hamiltonian system is given by a Hamiltonian
function $h: D^* \longrightarrow \R$. The dynamics is provided by
the corresponding Hamiltonian vector field ${\mathcal
H}_h^{\Lambda_{D^*}}$ (${\mathcal H}_h^{\Lambda_{D^*}}(f) = \{f,
h\}_{D^*}$, for all real function $f$ on $D^*$). Here,
$\Lambda_{D^*}$ is the almost Poisson tensor field defined from
$\{\cdot , \cdot \}_{D^*}$. The reader can immediately recognize
that we are extending the standard model, where $D=TQ$, $D^*=T^*Q$
and $\lcf \cdot , \cdot \rcf_{D}$ is the usual Lie bracket of
vector fields which is related with the canonical Poisson bracket
on $T^*Q$, so that $d^D$ is just, in this case, the usual exterior
differential. Another important fact is the introduction of the
notion of morphisms preserving the Hamiltonian system; indeed,
this concept will be very useful to treat with reduction
procedures for nonholonomic systems with symmetries. We remark
that this type of procedures were intensively discussed in the
seminal paper \cite{BlKrMaMu} by Bloch {\em et al.}

In the above framework we can prove the main result of our paper:
Theorem \ref{maintheorem}. In this theorem, we obtain the
Hamilton-Jacobi equation whose expression seems a natural
extension of the classical Hamilton-Jacobi equation for
unconstrained systems, as appears, for instance, in \cite{AbMa}.
Moreover,  our construction is preserved under the natural
morphims of the theory. This fact is proved in Theorem
\ref{morHa-Ja}.

Furthermore, using the orbit theorem (see \cite{AgSa}), we will
show that the classical form of the Hamilton-Jacobi equation:
$H\circ\alpha=\hbox{constant}$, with $\alpha: Q\to D^*$ satisfying
$d^D\alpha=0$,  remains valid for a special class of nonholonomic
mechanical systems: those satisfying the condition of being
completely nonholonomic. See Section \ref{orbit-theorem} for more
details and also the paper by Ohsawa and Bloch \cite{OhBl} for the
particular case when $D$ is a distribution on $Q$.

The above theorems are applied to the theory of mechanical systems
subjected to linear nonholonomic constraints on a Lie algebroid
$A$. The ingredients of this theory are a Lie algebroid $\tau_{A}:
A \to Q$ over a manifold $Q$, a Lagrangian function $L: A \to \R$
of mechanical type, and a vector subbundle $\tau_{D}: D \to Q$ of
$A$. The total space $D$ of this vector subbundle is the
constraint submanifold (see \cite{CoLeMaMa}). Then, using the
corresponding linear Poisson structure on $A^*$, one may introduce
a linear almost Poisson bracket on $D^*$, the so-called
nonholonomic bracket. A linear almost Poisson bracket on $D$ which
is isomorphic to the nonholonomic bracket was considered in
\cite{CoLeMaMa}; however, it should be remarked that our formalism
simplifies very much the procedure to obtain it. Using all these
ingredients one can apply the general procedure (Theorems
\ref{maintheorem} and \ref{morHa-Ja}) to obtain new and
interesting results. We also remark that the main part of the
relevant information for developing  the Hamilton-Jacobi equation
for the nonholonomic system $(L, D)$  is contained in the vector
subbundle $D$ or, equivalently, in its dual $D^*$ (see Theorems
\ref{maintheorem} and \ref{morHa-Ja}). Then, the computational
 cost is lower  than in
previous approximations to the theory.

In the particular case when $A$ is the standard Lie algebroid
$\tau_{TQ}: TQ\to Q$ then the constraint subbundle is  a
distribution $D$ on $Q$. The linear almost Poisson bracket on
$D^*$ is provided by the classical nonholonomic bracket (which is
usually induced from the canonical Poisson bracket on $T^*Q$),
clarifying previous constructions \cite{CaLeMa,IbLeMaMa,VaMa}. In
addition, as  a consequence, we recover some of the results
obtained in \cite{IgLeMa} about the Hamilton-Jacobi equation for
nonholonomic mechanical systems (see Corollary
\ref{nonHa-Jastan}). Moreover, we apply our results to an explicit
example: the  two-wheeled carriage. On the other hand, if our
Lagrangian system on an arbitrary Lie algebroid $A$ is
unconstrained (that is, the constraint subbundle $D = A$) then,
using our general theory, we recover some results on the
Hamilton-Jacobi equation for Lie algebroids (see Corollary
\ref{unconsHam-Ja}) which were proved in \cite{LeMaMa}.
Furthermore, if $A$ is the standard Lie algebroid $\tau_{TQ}: TQ
\to Q$ then we directly deduce some well-known facts about the
classical Hamilton-Jacobi equation (see Corollary
\ref{Ham-Jastan}).

Another interesting application is discussed; the particular case
when the Lie algebroid is the Atiyah algebroid $\tau_{\bar{A}}:
\bar{A}=TQ/G\to \bar{Q}=Q/G$ associated with a principal
$G$-bundle $F: Q\to \bar{Q}=Q/G$. In such a case, we have a
Lagrangian function $\bar{L}: \bar{A}\to \R$ of mechanical type
and a constraint subbundle $\tau_{\bar{D}}: \bar{D}\to \bar{Q}$ of
$\tau_{\bar{A}}: \bar{A}=TQ/G\to \bar{Q}$. This nonholonomic
system is precisely the reduction, in the sense of Theorem
\ref{morHa-Ja}, of a constrained system $(L, D)$ on the standard
Lie algebroid $\tau_{A}: A=TQ\to Q$. In fact, using Theorem
\ref{morHa-Ja}, we deduce that the solutions of the
Hamilton-Jacobi equations for both systems are related in a
natural way by projection. We also characterize the nonholonomic
bracket on $\bar{D}^*$. All these results are applied to a very
interesting example: the snakeboard. In this example, an explicit
expression of the reduced nonholonomic bracket is found; moreover,
the Hamilton-Jacobi equations are proposed and it is shown the
utility of our framework to integrate the equations of motion.

We expect that the results of this paper will be useful for
analytical integration of many difficult systems (see, as an
example, the detailed study of the snakeboard in this paper and
the examples in \cite{OhBl}) and the key for the construction of
geometric integrators based on the Hamilton-Jacobi equation (see,
for instance, Chapter VI in \cite{Hair} and references therein for
the particular case of standard nonholonomic mechanical systems).

The structure of the paper is as follows. In Section
\ref{section2}, the relation between linear almost Poisson
structures on a vector bundle, skew-symmetric algebroids and
almost differentials is obtained. In Section \ref{orbit-theorem},
we introduce the notion of a completely nonholonomic
skew-symmetric algebroid and we prove that on an algebroid $D$ of
this kind with connected base $Q$ the space $H^{0}(d^{D}) = \{f
\in C^{\infty}(Q) / d^D f = 0 \}$ is isomorphic to $\R$. We also
prove that on an arbitrary skew-symmetric algebroid $D$ the
condition $d^D f = 0$ implies that $f$ is constant on the leaves
of a certain generalized foliation (see Theorem
\ref{orb-theo-Lieal}). For this purpose, we will use the orbit
theorem. In Section \ref{section3}, we consider Hamiltonian
systems associated with a linear almost Poisson structure on the
dual bundle $D^*$ to a vector bundle and a Hamiltonian function on
$D^*$. Then, the Hamilton-Jacobi equation is proposed in this
setting. Moreover, using the results of Section
\ref{orbit-theorem}, we obtain an interesting expression of this
equation. In Section \ref{section4}, we apply the previous results
to nonholonomic mechanical systems and, in particular,  to some
explicit examples. Moreover, we review in this section some
previous approaches to the topic. We conclude our paper with the
future lines of work and an appendix with the proof of some
technical results.

\section{Linear almost Poisson structures, skew-symmetric algebroids
and almost differentials}\label{section2}

Most of the results contained in this section are well-known in
the literature (see \cite{GLMM,GrUr0,GrUr,Po,PoPo}). However, to make
the paper more self-contained, we will include their proofs.

Let $\tau_{D}: D \to Q$ be a vector bundle of rank $n$ over a
manifold $Q$ of dimension $m$. Denote by $D^*$ the dual vector
bundle to $D$ and by $\tau_{D^*}: D^* \to Q$ the corresponding
vector bundle projection.

\begin{definition}\label{laps}
A \emph{linear almost Poisson structure} on $D^*$ is a bracket of
functions
\[
\{ \cdot , \cdot \}_{D^*}: C^{\infty}(D^*) \times C^{\infty}(D^*)
\to C^{\infty}(D^*)
\]
such that:
\begin{enumerate}
\item
$\{ \cdot , \cdot \}_{D^*}$ is skew-symmetric, that is,
\[
\{\varphi, \psi \}_{D^*} = -\{\psi, \varphi\}_{D^*}, \; \; \mbox{
for } \varphi, \psi \in C^{\infty}(D^*).
\]

\item
$\{\cdot , \cdot \}_{D^*}$ satisfies the Leibniz rule, that is,
\[
\{\varphi \varphi', \psi \}_{D^*} = \varphi \{\varphi',
\psi\}_{D^*} + \varphi' \{\varphi, \psi \}_{D^*}, \; \; \mbox{ for
} \varphi, \varphi', \psi \in C^{\infty}(D^*).
\]

\item
$\{\cdot , \cdot \}_{D^*}$ is linear, that is, if $\varphi$ and
$\psi$ are linear functions on $D^*$ then $\{\varphi, \psi
\}_{D^*}$ is also a linear function.

\end{enumerate}

If, in addition, the bracket satisfies the Jacobi identity then we
have that $\{\cdot, \cdot\}_{D^*}$ is a \emph{linear Poisson
structure} on $D^*$.

\end{definition}

Properties (i) and (ii) in Definition \ref{laps} imply that there
exists a $2$-vector $\Lambda_{D^*}$ on $D^*$ such that
\[
\Lambda_{D^*}(d\varphi, d\psi) = \{\varphi, \psi\}_{D^*}, \; \;
\mbox{ for } \varphi, \psi \in C^{\infty}(D^*).
\]
$\Lambda_{D^*}$ is called the \emph{linear almost Poisson
$2$-vector} associated with the linear almost Poisson structure
$\{\cdot, \cdot\}_{D^*}$.

Note that there exists a one-to-one correspondence between the
space $\Gamma(\tau_{D})$ of sections of the vector bundle
$\tau_{D}: D \to Q$ and the space of linear functions on $D^{*}$.
In fact, if $X \in \Gamma(\tau_{D})$ then the corresponding linear
function $\hat{X}$ on $D^*$ is given by
\[
\hat{X}(\alpha) = \alpha(X(\tau_{D^*}(\alpha))), \; \; \mbox{ for
} \alpha \in D^*.
\]

\begin{proposition}\label{properlin}
Let $\{\cdot, \cdot\}_{D^*}$ be a linear almost Poisson structure
on $D^*$.
\begin{enumerate}
\item
If $X$ is a section of $\tau_{D}: D \to Q$ and $f$ is a real
$C^{\infty}$-function on $Q$ then the bracket $\{\hat{X}, f \circ
\tau_{D^*}\}_{D^*}$ is a basic function with respect to the
projection $\tau_{D^*}$.

\item
If $f$ and $g$ are real $C^{\infty}$-functions on $Q$ then
\[
\{f \circ \tau_{D^*}, g \circ \tau_{D^*}\}_{D^*} = 0.
\]
\end{enumerate}

\end{proposition}

\begin{proof}
Let $Y$ be an arbitrary section of $\tau_{D}: D \to Q$.

Using Definition \ref{laps}, we have that
\[
\{\hat{X}, (f \circ \tau_{D^*}) \hat{Y}\}_{D^*} = (f \circ
\tau_{D^*})\{\hat{X}, \hat{Y}\}_{D^{*}} + (\{\hat{X},
f\circ\tau_{D^*}\}_{D^*})\hat{Y}
\]
is a linear function on $D^*$. Thus, since $(f \circ
\tau_{D^*})\{\hat{X}, \hat{Y}\}_{D^*}$ also is a linear function,
it follows that $\{\hat{X}, f\circ \tau_{D^*}\}_{D^*}$ is a basic
function with respect to $\tau_{D^*}$. This proves (i).

On the other hand, using (i) and Definition \ref{laps}, we deduce
that
\[
\{(f \circ \tau_{D^*})\hat{Y}, g \circ \tau_{D^*}\}_{D^*} = (f
\circ \tau_{D^*})\{\hat{Y}, g\circ \tau_{D^*}\}_{D^*} + (\{f \circ
\tau_{D^*}, g\circ \tau_{D^*}\}_{D^*})\hat{Y}
\]
is a basic function with respect to $\tau_{D^*}$. Therefore, as
$(f \circ \tau_{D^*})\{\hat{Y}, g \circ \tau_{D^*}\}_{D^*}$ also
is a basic function with respect to $\tau_{D^*}$, we conclude that
$\{f \circ \tau_{D^*}, g \circ \tau_{D^*}\}_{D^*} = 0$. This
proves (ii).
\end{proof}

If $(q^i)$ are local coordinates on an open subset $U$ of $Q$ and
$\{X_{\alpha}\}$ is a basis of sections of the vector bundle
$\tau_{D}^{-1}(U) \to U$ then we have the corresponding local
coordinates $(q^i, p_{\alpha})$ on $D^*$. Moreover, from
Proposition \ref{properlin}, it follows that
\[
\{p_{\alpha}, p_{\beta}\}_{D^*} = -C_{\alpha
\beta}^{\gamma}p_{\gamma}, \; \; \{q^j, p_{\alpha}\}_{D^*} =
\rho^{j}_{\alpha}, \; \; \{q^i, q^j\}_{D^*} = 0,
\]
with $C_{\alpha \beta}^{\gamma}$ and $\rho^{j}_{\alpha}$ real
$C^{\infty}$-functions on $U$.

Consequently, the linear almost Poisson $2$-vector $\Lambda_{D^*}$
has the following local expression
\begin{equation}\label{LambdaDstar}
\Lambda_{D^*} = \rho^{j}_{\alpha} \displaystyle
\frac{\partial}{\partial q^j} \wedge \frac{\partial}{\partial
p_{\alpha}} - \frac{1}{2} C_{\alpha \beta}^{\gamma} p_{\gamma}
\frac{\partial}{\partial p_{\alpha}} \wedge
\frac{\partial}{\partial p_{\beta}}.
\end{equation}

\begin{definition}\label{alas}
An \emph{skew-symmetric algebroid structure} on the vector bundle
$\tau_{D}: D \to Q$ is a $\R$-linear bracket $\lcf \cdot, \cdot
\rcf_{D}: \Gamma(\tau_{D}) \times \Gamma(\tau_{D}) \to
\Gamma(\tau_{D})$ on the space $\Gamma(\tau_{D})$ and a vector
bundle morphism $\rho_{D}: D \to TQ$, the \emph{anchor map}, such
that:
\begin{enumerate}
\item
$\lcf \cdot, \cdot \rcf_{D}$ is skew-symmetric, that is,
\[
\lcf X, Y \rcf_{D} = -\lcf Y, X \rcf_{D}, \; \; \mbox{ for } X, Y
\in \Gamma(\tau_{D}).
\]

\item
If we also denote by $\rho_{D}: \Gamma(\tau_{D}) \to {\frak X}(Q)$
the morphism of $C^{\infty}(Q)$-modules induced by the anchor map
then
\[
\lcf X, fY \rcf_{D} = f \lcf X, Y \rcf_{D} + \rho_{D}(X)(f) Y, \;
\; \mbox{ for } X, Y \in \Gamma(D) \mbox{ and } f \in
C^{\infty}(Q).
\]

\end{enumerate}

If the bracket $\lcf \cdot, \cdot \rcf_{D}$ satisfies the Jacobi
identity, we have that the pair $(\lcf \cdot, \cdot \rcf_{D},
\rho_{D})$ is a \emph{Lie algebroid structure} on the vector
bundle $\tau_{D}: D \to Q$.

\end{definition}

\begin{remark}\label{foliation-remark} {\rm If $(D, \lcf\cdot, \cdot\rcf_{D}, \rho_{D})$ is
a Lie algebroid over $Q$, we may consider the generalized
distribution $\tilde{D}$ whose characteristic space at a point $q
\in Q$ is given by $\tilde{D}(q) = \rho_{D}(D_{q})$, where $D_{q}$
is the fibre of $D$ over $q$. The distribution $\tilde{D}$ is
finitely generated and involutive. Thus, $\tilde{D}$ defines a
generalized foliation on $Q$ in the sense of Sussmann \cite{Su}.
$\tilde{D}$ is the \emph{Lie algebroid foliation} on $Q$
associated with $D$. }
\end{remark}

Now, we will denote by ${\mathcal L}{\mathcal A}{\mathcal P}(D^*)$
(respectively, ${\mathcal L}{\mathcal P}(D^*)$) the set of linear
almost Poisson structures (respectively, linear Poisson
structures) on $D^*$. Denote also by ${\mathcal S}{\mathcal S}{\mathcal A}(D)$ (respectively, ${\mathcal L}{\mathcal A}(D)$)
the set of skew-symmetric algebroid (respectively, Lie algebroid)
structures on the vector bundle $\tau_{D}: D \to Q$. Then, we will
see in the next theorem that there exists a one-to-one
correspondence between ${\mathcal L}{\mathcal A}{\mathcal P}(D^*)$
(respectively, ${\mathcal L}{\mathcal P}(D^*)$) and the set of
skew-symmetric algebroid (respectively, Lie algebroid) structures
on $\tau_{D}: D \to Q$.

\begin{theorem}\label{alp-ala}
There exists a one-to-one correspondence $\Psi$ between the sets
${\mathcal L}{\mathcal A}{\mathcal P}(D^*)$ and ${\mathcal
S}{\mathcal S}{\mathcal A}(D)$. Under the bijection $\Psi$, the
subset ${\mathcal L}{\mathcal P}(D^*)$ of ${\mathcal L}{\mathcal
A}{\mathcal P}(D^*)$ corresponds with the subset ${\mathcal
L}{\mathcal A}(D)$ of ${\mathcal S}{\mathcal S}{\mathcal A}(D)$.
Moreover, if $\{\cdot, \cdot \}_{D^*}$ is a linear almost Poisson
structure on $D^*$ then the corresponding skew-symmetric algebroid
structure $(\lcf \cdot, \cdot \rcf_{D}, \rho_{D})$ on $D$ is
characterized by the following conditions
\begin{equation}\label{alafromalp}
\widehat{\lcf X, Y \rcf}_{D} = -\{\hat{X}, \hat{Y}\}_{D^*}, \; \;
\; \rho_{D}(X)(f) \circ \tau_{D^*} = \{f \circ \tau_{D^*},
\hat{X}\}_{D^*}
\end{equation}
for $X, Y \in \Gamma(\tau_{D})$ and $f \in C^{\infty}(Q)$.
\end{theorem}

\begin{proof}
Let $\{\cdot, \cdot\}_{D^*}$ be a linear almost Poisson structure
on $D^*$. Then, it is easy to prove that $\lcf \cdot, \cdot
\rcf_{D}$ (defined as in (\ref{alafromalp})) is a $\R$-bilinear
skew-symmetric bracket. Moreover, since $\{\cdot, \cdot\}_{D^*}$
satisfies the Leibniz rule, it follows that $\rho_{D}(X)$ is a
vector field on $Q$, for $X\in \Gamma(\tau_{D})$. In addition,
using again that $\{\cdot, \cdot\}_{D^*}$ satisfies the Leibniz
rule and Proposition \ref{properlin}, we deduce that
\[
\rho_{D}(gX) = g\rho_{D}(X), \; \; \; \mbox{ for } g \in
C^{\infty}(Q) \mbox{ and } X\in \Gamma(\tau_{D}).
\]
Thus, $\rho_{D}: \Gamma(\tau_{D}) \to {\frak X}(Q)$ is a morphism
of $C^{\infty}(Q)$-modules.

On the other hand, from (\ref{alafromalp}), we obtain that
\[
\widehat{\lcf X, fY \rcf}_{D} = -\{\hat{X}, (f\circ
\tau_{D^*})\hat{Y}\}_{D^*} = (\rho_{D}(X)(f) \circ
\tau_{D^*})\hat{Y} - (f \circ \tau_{D^*}) \{\hat{X},
\hat{Y}\}_{D^*}.
\]
Therefore,
\[
\lcf X, fY \rcf_{D} = f \lcf X, Y \rcf_{D} + \rho_{D}(X)(f) Y
\]
and $(\lcf \cdot, \cdot \rcf_{D}, \rho_{D})$ is an skew-symmetric
algebroid structure on $\tau_{D}: D \to Q$. It is clear that if
$\{\cdot, \cdot\}_{D^*}$ is a Poisson bracket then $\lcf \cdot,
\cdot \rcf_{D}$ satisfies the Jacobi identity.

Conversely, if $(\lcf \cdot, \cdot \rcf_{D}, \rho_{D})$ is an
skew-symmetric algebroid structure on $\tau_{D}: D \to Q$ and $x
\in Q$ then one may prove that there exists an open subset $U$ of
$Q$ and a unique linear almost Poisson structure on the vector
bundle $\tau_{\tau_{D}^{-1}(U)^*}: \tau_{D}^{-1}(U)^* \to U$ such
that
\[
\{\hat{X}, \hat{Y}\}_{\tau_{D}^{-1}(U)^*} = -\widehat{\lcf X,
Y\rcf}_{\tau_{D}^{-1}(U)}, \; \; \; \{f \circ
\tau_{\tau_{D}^{-1}(U)^*}, \hat{X}\}_{\tau_{D}^{-1}(U)^*} =
\rho_{\tau_{D}^{-1}(U)}(X)(f) \circ \tau_{\tau_{D}^{-1}(U)^*},
\]
and
\[
\{f \circ \tau_{\tau_{D}^{-1}(U)^*}, g \circ
\tau_{\tau_{D}^{-1}(U)^*}\}_{\tau_{D}^{-1}(U)^*} = 0,
\]
for $X, Y$ sections of the vector bundle $\tau_{D}^{-1}(U) \to U$
and $f, g \in C^{\infty}(U)$. Here, $(\lcf \cdot, \cdot
\rcf_{\tau_{D}^{-1}(U)}, \rho_{\tau_{D}^{-1}(U)})$ is the
skew-symmetric algebroid structure on $\tau_{D}^{-1}(U)$ induced,
in a natural way, by the skew-symmetric algebroid structure $(\lcf
\cdot, \cdot \rcf_{D}, \rho_{D})$ on $D$. In addition, we have
that if $\lcf \cdot, \cdot \rcf_{D}$ satisfies the Jacobi identity
then $\{\cdot, \cdot\}_{\tau_{D}^{-1}(U)^*}$ is a linear Poisson
bracket on $\tau_{D}^{-1}(U)^*$. Thus, we deduce that there exists
a unique linear almost Poisson structure $\{\cdot, \cdot\}_{D^*}$
on $D^*$ such that (\ref{alafromalp}) holds.
\end{proof}

Let $\{\cdot, \cdot\}_{D^*}$ be a linear almost Poisson structure
on $D^*$ and $(\lcf \cdot, \cdot \rcf_{D}, \rho_{D})$ be the
corresponding skew-symmetric algebroid structure on $\tau_{D}: D
\to Q$. If $(q^i)$ are local coordinates on an open subset $U$ of
$Q$ and $\{X_{\alpha}\}$ is a basis of sections of the vector
bundle $\tau_{D}^{-1}(U) \to U$ such that $\Lambda_{D^*}$ is given
by (\ref{LambdaDstar}) (on $\tau_{D}^{-1}(U)$) then
\[
\lcf X_{\alpha}, X_{\beta} \rcf_{D} = C_{\alpha \beta}^{\gamma}
X_{\gamma}, \; \; \; \rho_{D}(X_{\alpha}) = \rho_{\alpha}^{j}
\displaystyle \frac{\partial}{\partial q^j}.
\]
$C_{\alpha\beta}^{\gamma}$ and $\rho_{\alpha}^{j}$ are called the
\emph{local structure functions} of the skew-symmetric algebroid
structure $(\lcf \cdot, \cdot \rcf_{D}, \rho_{D})$ with respect to
the local coordinates $(q^i)$ and the basis $\{X_{\alpha}\}$.

Next, we will see that there exists a one-to-one correspondence
between ${\mathcal S}{\mathcal S}{\mathcal A}(D)$ and the set of
almost differentials on the vector bundle $\tau_{D}: D \to M$.

\begin{definition}\label{almostdiff}
\emph{An almost differential} on the vector bundle $\tau_{D}: D
\to Q$ is a $\R$-linear map
\[
d^{D}: \Gamma(\Lambda^{k}\tau_{D^*}) \to
\Gamma(\Lambda^{k+1}\tau_{D^*}), \; \; k\in \{0, \dots, n-1\}
\]
such that
\begin{equation}\label{Derivation}
d^{D}(\alpha \wedge \beta) = d^{D}\alpha \wedge \beta + (-1)^{k}
\alpha \wedge d^{D}\beta, \; \; \mbox{ for } \alpha \in
\Gamma(\Lambda^{k}\tau_{D^*}) \mbox{ and } \beta \in
\Gamma(\Lambda^{r}\tau_{D^*}).
\end{equation}

If $(d^{D})^{2} = 0$ then $d^{D}$ is said to be \emph{a
differential} on the vector bundle $\tau_{D}: D \to Q$.

\end{definition}

Denote by ${\mathcal A}{\mathcal D}(D)$ (respectively, ${\mathcal
D}(D)$) the set of almost differentials (respectively,
differentials) on the vector bundle $\tau_{D}: D \to Q$.

\begin{theorem}\label{ala-ad}
There exists a one-to-one correspondence $\Phi$ between the sets
${\mathcal S}{\mathcal S}{\mathcal A}(D)$ and ${\mathcal
A}{\mathcal D}(D)$. Under the bijection $\Phi$ the subset
${\mathcal L}{\mathcal A}(D)$ of ${\mathcal S}{\mathcal S}{\mathcal A}(D)$ corresponds with the subset ${\mathcal D}(D)$
of ${\mathcal A}{\mathcal D}(D)$. Moreover, we have:
\begin{enumerate}
\item
If $d^{D}$ is an almost differential on the vector bundle
$\tau_{D}: D \to Q$ then the corresponding skew-symmetric
algebroid structure $(\lcf \cdot, \cdot \rcf_{D}, \rho_{D})$ on
$D$ is characterized by the following conditions:
\begin{equation}\label{bracDrhoD}
\alpha(\lcf X, Y \rcf_{D}) = d^{D}(\alpha(Y))(X) -
d^{D}(\alpha(X))(Y) - (d^{D}\alpha)(X, Y), \; \; \rho_{D}(X)(f) =
(d^{D}f)(X),
\end{equation}
for $X, Y \in \Gamma(\tau_{D})$, $\alpha \in \Gamma(\tau_{D^*})$
and $f \in C^{\infty}(Q)$.

\item
If $(\lcf \cdot, \cdot \rcf_{D}, \rho_{D})$ is an skew-symmetric
algebroid structure on the vector bundle $\tau_{D}: D \to Q$ then
the corresponding almost differential $d^{D}$ is defined by
\begin{equation}\label{deD}
\begin{array}{rcl}
(d^{D}\alpha)(X_{0}, X_{1}, \dots, X_{k}) &=& \displaystyle
\sum_{i=0}^{k} (-1)^{i} \rho_{D}(X_{i})(\alpha(X_{0}, \dots,
\hat{X}_{i}, \dots, X_{k})) \\[5pt]
&& + \displaystyle \sum_{i < j} (-1)^{i+j} \alpha (\lcf X_{i},
X_{j} \rcf_{D}, X_{0}, X_{1}, \dots, \hat{X}_{i}, \dots,
\hat{X}_{j}, \dots, X_{k})
\end{array}
\end{equation}
for $\alpha \in \Gamma(\Lambda^{k}\tau_{D^*})$ and $X_{0}, \dots,
X_{k} \in \Gamma(\tau_{D})$.
\end{enumerate}
\end{theorem}
\begin{proof}
Let $d^{D}$ be an almost differential on $\tau_{D}: D \to Q$ and
suppose that $\rho_{D}$ and $\lcf \cdot, \cdot \rcf_{D}$ are given
by (\ref{bracDrhoD}). Then, using the fact that
\[
d^{D}(ff') = fd^{D}f' + f'd^{D}f, \; \; \; \mbox{ for } f, f' \in
C^{\infty}(Q),
\]
we deduce that $(\lcf \cdot, \cdot \rcf_{D}, \rho_{D})$ is an
skew-symmetric algebroid structure on $\tau_{D}: D \to Q$.
Moreover, it is well-known that if $(d^{D})^2 = 0$ then $\lcf
\cdot, \cdot \rcf_{D}$ satisfies the Jacobi identity and the pair
$(\lcf \cdot, \cdot \rcf_{D}, \rho_{D})$ is a Lie algebroid
structure on $\tau_{D}: D \to Q$ (see, for instance,
\cite{KoMa,Xu}).

Conversely, if $(\lcf \cdot, \cdot \rcf_{D}, \rho_{D})$ is an
skew-symmetric algebroid structure on $\tau_{D}: D \to Q$ and
$d^{D}$ is the operator defined by (\ref{deD}) then, it is clear
that,
\[
(d^{D}\alpha)(X, Y) = \rho_{D}(X)(\alpha(Y)) -
\rho_{D}(Y)(\alpha(X)) - \alpha \lcf X, Y \rcf_{D}, \; \; (d^{D}
f)(X) = \rho_{D}(X)(f),
\]
for $f \in C^{\infty}(M)$, $\alpha \in \Gamma(\tau_{D^*})$ and $X,
Y \in \Gamma(\tau_{D})$. In addition, an straightforward
computation proves that
\[
d^{D}(\alpha \wedge \beta) = d^{D}\alpha \wedge \beta + (-1)^{k}
\alpha \wedge d^{D}\beta, \; \; \mbox{ for } \alpha \in
\Gamma(\Lambda^{k}\tau_{D^*}) \mbox{ and } \beta \in
\Gamma(\Lambda^{r}\tau_{D^*}).
\]
Finally, it is well-known that if $\lcf \cdot, \cdot \rcf_{D}$
satisfies the Jacobi identity then $(d^{D})^{2} = 0$ (see, for
instance, \cite{Mac}).
\end{proof}

Let $(\lcf \cdot, \cdot \rcf_{D}, \rho_{D})$ be a skew-symmetric
algebroid structure on the vector bundle $\tau_{D}: D \to Q$ and
$d^{D}$ be the corresponding almost differential. If $(q^i)$ are
local coordinates on an open subset $U$ of $Q$ and
$\{X_{\alpha}\}$ is a basis of sections of the vector bundle
$\tau_{D}^{-1}(U) \to U$ such that $C_{\alpha\beta}^{\gamma}$ and
$\rho_{\alpha}^{j}$ are the local structure functions of the
skew-symmetric algebroid structure, then
\[
d^{D}x^i = \rho_{\alpha}^{i}X^{\alpha}, \; \; d^{D}X^{\gamma} =
\displaystyle -\frac{1}{2} C_{\alpha\beta}^{\gamma}X^{\alpha}
\wedge X^{\beta},
\]
for all $i$ and $\gamma$.

{}From Theorems \ref{alp-ala} and \ref{ala-ad}, we conclude the
following result
\begin{theorem}
Let $\tau_{D}: D \to Q$ be a vector bundle over a manifold $Q$ and
$D^*$ be its dual vector bundle. Then, there exists a one-to-one
correspondence between the set ${\mathcal L}{\mathcal A}{\mathcal
P}(D^*)$ of linear almost Poisson structures on $D^*$, the set
${\mathcal S}{\mathcal S}{\mathcal A}(D)$ of skew-symmetric
algebroid structures on $\tau_{D}: D \to Q$ and the set ${\mathcal
A}{\mathcal D}(D)$ of almost differentials on this vector bundle.
\end{theorem}

\section{Skew-symmetric algebroids and the orbit theorem}\label{orbit-theorem}

Let $(D, \lcf \cdot , \cdot \rcf_{D}, \rho_{D})$ be a
skew-symmetric algebroid over $Q$ and $d^D$ be the corresponding
almost differential.

We can consider the vector space over $\R$
\[
H^0(d^D) = \{f \in C^{\infty}(Q) / d^D f = 0\}.
\]
Note that if $D$ is a Lie algebroid we have that $H^{0}(d^D)$ is
the $0$ Lie algebroid cohomology group associated with $D$.

On the other hand, it is clear that if $Q$ is connected and $D$ is
a transitive skew-symmetric algebroid, that is,
\[
\rho_D(D_{q}) = T_{q}Q, \; \; \mbox{ for all } q\in Q
\]
then
\begin{equation}\label{Poincare}
H^0(d^D) \simeq \R.
\end{equation}
Condition (\ref{Poincare}) will play an important role in Section
\ref{subsection3.1}.

Next, we will see that (\ref{Poincare}) holds if the
skew-symmetric algebroid is completely nonholonomic with connected
base space.

Let $\tilde{D}$ be the generalized distribution on $Q$ whose
characteristic space at the point $q\in Q$ is
\[
\tilde{D}_{q} = \rho_{D}(D_{q}).
\]
It is clear that $\tilde{D}$ is finitely generated. Note that the
$C^{\infty}$-module $\Gamma(D)$ is finitely generated (see
\cite{GrHaVa}).

Now, denote by $\mathrm{Lie}^{\infty}(\tilde{D})$ the smallest Lie
subalgebra of ${\frak X}(Q)$ containing $\tilde{D}$. Then
$\mathrm{Lie}^{\infty}(\tilde{D})$ is comprised of finite
$\R$-linear combinations of vector fields of the form
\[
[\tilde{X}_{k},[\tilde{X}_{k-1}, \dots [\tilde{X}_{2},
\tilde{X}_{1}]\dots ]]
\]
with $k\in \mathbb{N}$, $k\neq 0$, and $\tilde{X}_{1}, \dots ,
\tilde{X}_{k} \in {\frak X}(Q)$ satisfying
\[
\tilde{X}_{l}(q) \in \tilde{D}_{q}, \; \; \mbox{ for all } q\in Q
\]
(see \cite{AgSa}).

For each $q\in Q$, we will consider the vector subspace $\mathrm{Lie}_q^{\infty}(\tilde{D})$ of $T_{q}Q$ given by
\[
\mathrm{Lie}_{q}^{\infty}(\tilde{D}) = \{\tilde{X}(q) \in T_{q}Q /
\tilde{X} \in \mathrm{Lie}^{\infty}(\tilde{D})\}.
\]
Then, the assignment
\[
q \in Q \to \mathrm{Lie}_{q}^{\infty}(\tilde{D}) \subseteq T_{q}Q
\]
defines a generalized foliation on $Q$. The leaf $L$ of this
foliation over the point $q_{0} \in Q$ is \emph{the orbit of
$\tilde{D}$} over the point $q_{0}$, that is,
\[
L = \{(\phi_{\tilde{t}_{k}}^{\tilde{X}_{k}} \circ \dots \circ
\phi_{\tilde{t}_{1}}^{\tilde{X}_{1}})(q_{0}) \in Q / \tilde{t}_{l}
\in \R, \tilde{X}_{l} \in {\frak X}(Q) \mbox{ and }
\tilde{X}_{l}(q) \in \tilde{D}_{q}, \mbox{ for all } q\in Q \}.
\]
Here, $\phi_{\tilde{t}_{l}}^{\tilde{X}_{l}}$ is the flow of the
vector field $\tilde{X}_{l}$ at the time $\tilde{t}_{l}$ (for more
details, see \cite{AgSa}).

\begin{definition}\label{complete-nonholo}
The skew-symmetric algebroid $(D, \lcf \cdot, \cdot \rcf_{D},
\rho_{D})$ over $Q$ is said to be \emph{completely nonholonomic}
if
\[
\mathrm{Lie}_{q}^{\infty}(\rho_{D}(D)) = \mathrm{Lie}_{q}^{\infty}(\tilde{D}) = T_{q}Q, \; \; \mbox{ for all } q\in
Q.
\]
\end{definition}
Thus, if $Q$ is a connected manifold, it follows that $D$ is
completely nonholonomic if and only if the orbit of $\tilde{D}$
over any point $q_{0} \in Q$ is $Q$.

\begin{remark}{\rm
\begin{enumerate}
\item
Definition \ref{complete-nonholo} may be extended for anchored
vector bundles. A vector bundle $\tau_{D}: D \to Q$ over $Q$ is
said to be \emph{anchored} if it admits an anchor map, that is, a
vector bundle morphism $\rho_{D}: D \to TQ$. In such a case, the
vector bundle is said to be completely nonholonomic if $\mathrm{Lie}_{q}^{\infty}(\rho_{D}(D)) = \mathrm{Lie}_{q}^{\infty}(\tilde{D}) = T_{q}Q$, for all $q\in Q$.
\item
If $D$ is a regular distribution on a manifold $Q$ then the
inclusion map $i_{D}: D \to TQ$ is an anchor map for the vector
bundle $\tau_{D}: D \to Q$. Moreover, the anchored vector bundle
$\tau_{D}: D \to Q$ is completely nonholonomic if and only if the
distribution $D$ is completely nonholonomic in the classical sense
of Vershik and Gershkovich \cite{VeGe}.  In this sense it is
formulated in  the literature the classical \emph{Rashevsky-Chow
theorem}: If $\mathrm{Lie}_q^{\infty}(D)=T_qQ$, for all $q\in Q$,
then each orbit is equal to the whole manifold $Q$.
\end{enumerate}
}
\end{remark}
Now, we deduce the following result
\begin{proposition}
If the skew-symmetric algebroid $(D, \lcf \cdot, \cdot \rcf_{D},
\rho_{D})$ over $Q$ is completely nonholonomic and $Q$ is
connected then $H^{0}(d^D) \simeq \R$.
\end{proposition}
\begin{proof}
Suppose that $f \in C^{\infty}(Q)$ and that $d^{D}f = 0$. Let
$q_{0}$ be a point of $Q$. We must prove that
\[
\tilde{v}(f) = 0, \; \; \mbox{ for all } \tilde{v} \in T_{q_{0}}Q.
\]
The condition $(d^Df)(q_{0}) = 0$ implies that $\tilde{v}(f) = 0$,
for all $\tilde{v} \in \tilde{D}_{q_{0}}$.

Thus, first we have that
\[
0 = \tilde{X}_{1}(\tilde{X}_{2}(f)) -
\tilde{X}_{2}(\tilde{X}_{1}(f)) = [\tilde{X}_{1},
\tilde{X}_{2}](f),
\]
for $\tilde{X}_{1}, \tilde{X}_{2} \in {\frak X}(Q)$ and
$\tilde{X}_{1}(q), \tilde{X}_{2}(q) \in \tilde{D}_{q}$, for all
$q\in Q$.

Secondly, since $D$ is completely nonholonomic then
$\mathrm{Lie}^{\infty}_{q_{0}}(\tilde{D}) = T_{q_{0}}Q$.
Therefore, there exists a finite sequence of vector fields on $Q$,
$\tilde{X}_{1}, \dots , \tilde{X}_{k}$ such that $\tilde{X}_{i}(q)
\in \tilde{D}_{q}$, for all $i \in \{1, \dots , k\}$ and $q \in
Q$,
\[
\tilde{v} = [\tilde{X}_{k},[\tilde{X}_{k-1}, \dots [\tilde{X}_{2},
\tilde{X}_{1}] \dots ]](q_{0}).
\]
{}From both considerations, we deduce the result.
\end{proof}
However, the condition $H^{0}(d^{D}) \simeq \R$ does not imply, in
general, that the skew-symmetric algebroid $D$ is completely
nonholonomic.

In fact, let $D$ be the tangent bundle to $\R^2$
\[
\tau_{T\R^2}: T\R^2 \to \R^2.
\]
If $(x, y)$ are the standard coordinates on $\R^2$, it follows
that $\{X_1 = \displaystyle \frac{\partial}{\partial x}, X_{2}=
\displaystyle \frac{\partial}{\partial y}\}$ is a global basis of
$\Gamma(T\R^2) = {\frak X}(\R^2)$. So, we can consider the
skew-symmetric algebroid structure $(\lcf \cdot, \cdot
\rcf_{T\R^2}, \rho_{T\R^2})$ on $T\R^2$ which is characterized by
the following conditions
\[
\lcf X_{1}, X_{2}\rcf_{T\R^2} = 0, \; \; \rho_{T\R^2}(X_{1}) =
\displaystyle \frac{\partial}{\partial x}, \; \;
\rho_{T\R^2}(X_{2}) = xy\frac{\partial}{\partial y}.
\]
Then, the generalized distribution $\tilde{D} = \widetilde{T\R^2}$
on $\R^2$ is generated by the vector fields
\[
\tilde{X}_{1} = \displaystyle \frac{\partial}{\partial x}, \; \;
\tilde{X}_{2} = xy \displaystyle \frac{\partial}{\partial y}.
\]
Thus, the Lie subalgebra $\mathrm{Lie}^{\infty}(\tilde{D})$ of
$T\R^2$ is generated by the vector fields
\[
\tilde{X}_{1} = \displaystyle \frac{\partial}{\partial x}, \; \
\tilde{X}_{2} = \displaystyle xy \frac{\partial}{\partial y}, \;
\; \tilde{X}_{3} = \displaystyle y \frac{\partial}{\partial y}.
\]
This implies that
\[
\mathrm{Lie}_{(x_0, y_0)}^{\infty}(\tilde{D}) \neq T_{(x_{0},
y_{0})}\R^2, \; \; \mbox{ if } y_{0} = 0.
\]
However, if $f \in C^{\infty}(\R^2)$ and $d^Df = 0$, we deduce
that
\[
\displaystyle \frac{\partial f}{\partial x} = 0, \; \; \;
\displaystyle xy\frac{\partial f}{\partial y} = 0.
\]
Consequently, using that $f \in C^{\infty}(\R^2)$, we obtain that
$f$ is constant.

Next, we will discuss the case when the generalized foliation
$\mathrm{Lie}^{\infty}(\tilde{D}) \neq TQ$. In fact, we will prove
the following result.
\begin{theorem}\label{orb-theo-Lieal}
Let $(D, \lcf \cdot , \cdot \rcf_{D}, \rho_{D})$ be a
skew-symmetric algebroid over a manifold $Q$ and $f$ be a real
$C^{\infty}$-function on $Q$ such that $d^D f = 0$. Suppose that
$L$ is an orbit of $\tilde{D}$ and that $D_{L}$ is the vector
bundle over $L$ given by $D_{L} = \cup_{q\in L}D_{q} =
\tau_{D}^{-1}(L)$. Then:
\begin{enumerate}
\item
The couple $(\lcf \cdot, \cdot \rcf_{D}, \rho_{D})$ induces a
skew-symmetric algebroid structure $(\lcf \cdot , \cdot
\rcf_{D_{L}}, \rho_{D_{L}})$ on the vector bundle $\tau_{D_{L}}:
D_{L} \to L$ and the skew-symmetric algebroid $(D_{L}, \lcf \cdot
, \cdot \rcf_{D_{L}}, \rho_{D_{L}})$ is completely nonholonomic.
\item
The restriction of $f$ to $L$ is constant.
\end{enumerate}
\end{theorem}
\begin{proof}
It is clear that
\[
\rho_{D}(D_{q}) = \tilde{D}_{q} \subseteq \mathrm{Lie}_{q}^{\infty}(\tilde{D}) = T_{q}L, \; \; \mbox{ for all } q\in
L.
\]
Thus, we have a vector bundle morphism
\[
\rho_{D_{L}}: D_{L} \to TL.
\]
More precisely, $\rho_{D_{L}} = (\rho_{D})_{|D_{L}}$.

On the other hand, we may define a $\R$-bilinear skew-symmetric
bracket
\[
\lcf \cdot , \cdot \rcf_{D_{L}}: \Gamma(D_{L}) \times
\Gamma(D_{L}) \to \Gamma(D_{L}).
\]
In fact, if $X_{L}, Y_{L} \in \Gamma(D_{L})$ then
\[
\lcf X_{L}, Y_{L} \rcf_{D_{L}}(q) = \lcf X, Y \rcf_{D}(q), \; \;
\mbox{ for all } q \in L
\]
where $X, Y$ are sections of $\tau_{D}: D \to Q$ such that
\[
X_{|U\cap L} = (X_{L})_{|U\cap L}, \; \; Y_{|U\cap L} =
(Y_{L})_{|U\cap L},
\]
with $U$ an open subset of $Q$ and $q\in U$.

Note that the condition
\[
\mathrm{Lie}_{q}^{\infty}(\tilde{D}) = T_{q}L, \; \; \mbox{ for
all } q\in L,
\]
implies that
\[
(\lcf X, Y \rcf_{D})_{|V\cap L} = 0
\]
for $X, Y \in \Gamma(D)$, with $V$ an open subset of $Q$ and
\[
X(q) = 0, \; \; \mbox{ for all } q \in V \cap L.
\]
Therefore, the map $\lcf \cdot, \cdot \rcf_{D_{L}}$ is
well-defined. We remark that if $q \in V\cap L$ and
$\{X_{\alpha}\}$ is a local basis of $\Gamma(\tau_{D})$ in an open
susbset $W$ of $Q$, with $q\in W$, such that $X=
f^{\alpha}X_{\alpha}$ in $W$ then, using that
$(f^{\alpha})_{|V\cap W \cap L} = 0$ and that $\rho_{D}(Y)(q) \in
T_{q}L$, we deduce that
\[
\lcf X, Y \rcf_{D}(q) = f^{\alpha}(q) \lcf X_{\alpha}, Y \rcf_{D}
(q) - \rho_{D}(Y)(q)(f^{\alpha}) X_{\alpha}(q) = 0.
\]

 Moreover, the couple $(\lcf \cdot , \cdot \rcf_{D_{L}},
\rho_{D_{L}})$ is a skew-symmetric algebroid structure on the
vector bundle $\tau_{D_{L}}: D_{L} \to L$.

In addition, it is clear that
\[
(\widetilde{D_{L}})_{q} = \rho_{D_{L}}((D_{L})_{q}) =
\tilde{D}_{q}, \; \; \;
\mathrm{Lie}_{q}^{\infty}(\widetilde{D_{L}}) =
\mathrm{Lie}_{q}^{\infty}(\tilde{D})
\]
for all $q\in L$. Thus, we have that the skew-symmetric algebroid
$(D_{L}, \lcf \cdot , \cdot \rcf_{D_{L}}, \rho_{D_{L}})$ is
completely nonholonomic.

This proves (i).

On the other hand, it follows that the condition $d^D f = 0$
implies that
\[
d^{D_{L}}(f_{|L}) = 0
\]
and, since $H^0(d^{D_{L}}) \simeq \R$ (as a consequence of the
first part of the theorem), we conclude that $f$ is constant on
$L$.
\end{proof}

It will be also interesting to characterize under what conditions
there exist functions $f\in C^{\infty}(Q)$ such that $(d^D)^2
f=0$. Using Equations (\ref{bracDrhoD}) we easily deduce that:
\[
\left((d^D)^2 f\right)(X,Y)=\left(
[\rho_D(X),\rho_D(Y)]-\rho_D\lcf X, Y\rcf_D\right) f
\]
for all $X, Y\in \Gamma(\tau_D)$. Now, consider the generalized
distribution $\bar{D}$ on $Q$ whose characteristic space
$\bar{D}_{q}$ at the point $q\in Q$ is:
\[
\bar{D}_{q}=\left\{ ([\rho_D(X),\rho_D(Y)]-\rho_D\lcf X,
Y\rcf_D)(q) \;/\; X, Y\in \Gamma(\tau_D) \right\}.
\]
It is also clear that $\bar{D}$ is finitely generated. Denote by
$\mathrm{Lie}^{\infty}(\bar{D})$ the smallest Lie subalgebra of
${\frak X}(Q)$ containing $\bar{D}$. Observe that
$\mathrm{Lie}^{\infty}(\bar{D}) \subseteq
\mathrm{Lie}^{\infty}(\tilde{D})$. We deduce that  $(d^D)^2 f=0$
if and only if $f$ is constant on any orbit $\bar{L}$ of
$\bar{D}$. Of course if $\bar{D}$ is completely nonholonomic then
the unique functions $f\in C^{\infty}(Q)$ satisfying  $(d^D)^2
f=0$ are $f=\hbox{constant}$, but it has not to be always the case
and it may useful to find this particular type of functions in
concrete examples.

\section{Linear almost Poisson structures and Hamilton-Jacobi
equation}\label{section3}

\subsection{Linear almost Poisson structures and Hamiltonian
systems}\label{subsection3.1}

In this section, we will consider Hamiltonian systems associated
with a linear almost Poisson structure on the dual bundle $D^*$ to
a vector bundle and with a Hamiltonian function on $D^*$. Thus,
the ingredients of our theory are:
\begin{enumerate}
\item
A vector bundle $\tau_{D}: D \to Q$ of rank $n$ over a manifold
$Q$ of dimension $m$;
\item
A linear almost Poisson structure $\{\cdot, \cdot\}_{D^*}$ on
$D^*$ and
\item
A Hamiltonian function $h: D^* \to \R$ on $D^*$.
\end{enumerate}
The triplet $(D, \{\cdot, \cdot\}_{D^*}, h)$ is said to be \emph{a
Hamiltonian system}.

We will denote by $\Lambda_{D^*}$ the linear almost Poisson
$2$-vector associated with $\{\cdot, \cdot\}_{D^*}$. Then, we may
introduce the vector field ${\mathcal H}_{h}^{\Lambda_{D^*}}$ on
$D^*$ given by
\[
{\mathcal H}_{h}^{\Lambda_{D^*}} = -i(dh)\Lambda_{D^*}.
\]

${\mathcal H}_{h}^{\Lambda_{D^*}}$ is called \emph{the Hamiltonian
vector field} of $h$ with respect to $\Lambda_{D^*}$. The integral
curves of ${\mathcal H}_{h}^{\Lambda_{D^*}}$ are the solutions of
\emph{the Hamilton equations} for $h$.

Now, suppose that $(q^{i})$ are local coordinates on an open
subset $U$ of $Q$ and that $\{X_{\alpha}\}$ is a basis of the
space of sections of the vector bundle $\tau_{D}^{-1}(U) \to U$.
Denote by $(q^i, p_{\alpha})$ the corresponding local coordinates
on $D^*$ and by $C_{\alpha \beta}^{\gamma}$ and $\rho^i_{\alpha}$
the local structure functions (with respect to the coordinates
$(q^j)$ and to the basis $\{X_{\alpha}\}$) of the corresponding
skew-symmetric algebroid structure on $D$. Then, using
(\ref{LambdaDstar}), it follows that
\begin{equation}\label{HhLambda}
{\mathcal H}^{\Lambda_{D^*}}_{h} = \rho^i_{\alpha} \displaystyle
\frac{\partial h}{\partial p_{\alpha}} \frac{\partial}{\partial
q^{i}} - (\rho^{i}_{\alpha} \frac{\partial h}{\partial q^i} +
C_{\alpha\beta}^{\gamma}p_{\gamma} \frac{\partial h}{\partial
p_{\beta}}) \frac{\partial}{\partial p_{\alpha}}.
\end{equation}
Therefore, the Hamilton equations are
\[
\displaystyle \frac{dq^i}{dt} = \rho^i_{\alpha} \frac{\partial
h}{\partial p_{\alpha}}, \; \; \; \frac{dp_{\alpha}}{dt} = -
(\rho^{i}_{\alpha} \frac{\partial h}{\partial q^i} +
C_{\alpha\beta}^{\gamma}p_{\gamma} \frac{\partial h}{\partial
p_{\beta}}).
\]

\subsection{Hamiltonian systems and Hamilton-Jacobi equation}

Let $(D, \{\cdot, \cdot\}_{D^*}, h)$ be a Hamiltonian system and
$\alpha: Q \to D^*$ be a section of the vector bundle $\tau_{D^*}:
D^* \to Q$.

If ${\mathcal H}_{h}^{\Lambda_{D^*}}$ is the Hamiltonian vector
field of $h$ with respect to $\{\cdot, \cdot\}_{D^*}$, we may
introduce the vector field ${\mathcal
H}_{h,\alpha}^{\Lambda_{D^*}}$ on $Q$ given by
\[
{\mathcal H}_{h,\alpha}^{\Lambda_{D^*}}(q) =
(T_{\alpha(q)}\tau_{D^ *})({\mathcal
H}_{h}^{\Lambda_{D^*}}(\alpha(q))), \; \; \mbox{ for } q\in Q.
\]
{} From (\ref{HhLambda}), it follows that
\begin{equation}\label{HhFolcharac}
{\mathcal H}_{h,\alpha}^{\Lambda_{D^*}}(q) \in \rho_{D}(D_{q}), \;
\; \mbox{ for all } q \in Q,
\end{equation}
where $(\lcf \cdot, \cdot \rcf_{D}, \rho_{D})$ is the induced
skew-symmetric algebroid structure on the vector bundle $\tau_{D}:
D \to Q$.

Then, the aim of this section is to prove the following result.

\begin{theorem}\label{maintheorem}
Let $(D, \{\cdot, \cdot\}_{D^*}, h)$ be a Hamiltonian system and
$\alpha: Q \to D^*$ be a section of the vector bundle $\tau_{D^*}:
D^* \to Q$ such that $d^D\alpha = 0$. Under these hypotheses, the
following conditions are equivalent:
\begin{enumerate}
\item
If $c: I \to Q$ is an integral curve of the vector field
${\mathcal H}_{h,\alpha}^{\Lambda_{D^*}}$, that is,
\[
\dot{c}(t) = (T_{\alpha(c(t))}\tau_{D^*})({\mathcal
H}_{h}^{\Lambda_{D^*}}(\alpha(c(t)))), \; \; \mbox{ for all } t\in
I,
\]
then $\alpha \circ c: I \to D^*$ is a solution of the Hamilton
equations for $h$.

\item
$\alpha$ satisfies \emph{the Hamilton-Jacobi equation}
\[
d^{D}(h \circ \alpha) = 0.
\]
\end{enumerate}
\end{theorem}
\begin{remark}{\rm Let $\tilde{D}$ be the generalized distribution
on $Q$ given by $\tilde{D} = \rho_{D}(D)$ and $\mathrm{Lie}^{\infty}(\tilde{D})$ be the smallest Lie subalgebra of ${\frak
X}(Q)$ containing $\tilde{D}$. Then, using Theorem
\ref{orb-theo-Lieal}, we deduce that the Hamilton-Jacobi equation
holds for the section $\alpha$ if and only if the function $h\circ
\alpha: Q \to \R$ is constant on the leaves of the foliation
$\mathrm{Lie}^{\infty}(\tilde{D})$.}
\end{remark}

In order to prove Theorem \ref{maintheorem}, we will need some
previous results:

\begin{proposition}\label{primerresultado}
Let $\{\cdot, \cdot\}_{D^*}$ be a linear almost Poisson structure
on the dual bundle $D^*$ to a vector bundle $\tau_{D}: D \to Q$,
$\alpha: Q \to D^*$ be a section of $\tau_{D^*}: D^* \to Q$ and
$\#_{\Lambda_{D^*}}: T^*D^* \to TD^*$ be the vector bundle
morphism between $T^*D^*$ and $TD^*$ induced by the linear almost
Poisson $2$-vector $\Lambda_{D^*}$. Then, $\alpha$ is a
$1$-cocycle with respect to $d^D$ (i.e., $d^D\alpha = 0$) if and
only if for every point $q$ of $Q$ the subspace of
$T_{\alpha(q)}D^*$
\begin{equation}\label{DefLalD}
{\mathcal L}_{\alpha, D}(q) = (T_{q}\alpha)(\rho_{D}(D_{q}))
\end{equation}
is Lagrangian with respect to $\Lambda_{D^*}$, that is,
\[
\#_{\Lambda_{D^*}}(({\mathcal L}_{\alpha, D}(q))^0) = {\mathcal
L}_{\alpha, D}(q), \; \; \mbox{ for all } q \in Q.
\]
\end{proposition}

\begin{remark}{\rm
If $D = TQ$ and $\{\cdot, \cdot\}_{T^*Q}$ is the canonical Poisson
(symplectic) structure on $T^*Q$ then $\alpha$ is a $1$-form on
$Q$, $d^D = d^{TQ}$ is the standard exterior differential on $Q$
and $\rho_{D} = \rho_{TQ}: TQ \to TQ$ is the identity map. Thus,
if we apply Proposition \ref{primerresultado} we obtain that
$\alpha$ is a closed $1$-form if and only if $\alpha(Q)$ is a
Lagrangian submanifold of $T^*Q$. This is a well-known result in
the literature (see, for instance, \cite{AbMa}). }
\end{remark}

\begin{proposition}\label{segundoresultado}
Under the same hypotheses as in Proposition \ref{primerresultado},
if the section $\alpha$ is a $1$-cocycle with respect to $d^D$,
i.e., $d^D\alpha = 0$ then we have that
\[
Ker (\#_{\Lambda_{D^*}}(\alpha(q))) \subseteq ({\mathcal
L}_{\alpha, D}(q))^{0}, \; \; \; \mbox{ for all } q \in Q.
\]
\end{proposition}
The proofs of Propositions \ref{primerresultado} and
\ref{segundoresultado} may be found in the Appendix of this paper.

\noindent {\it Proof of Theorem \ref{maintheorem}.} It is clear that condition
(i) in Theorem \ref{maintheorem} is equivalent to the following
fact:

(i') {\it The vector fields ${\mathcal
H}_{h,\alpha}^{\Lambda_{D^*}}$ and ${\mathcal
H}_{h}^{\Lambda_{D^*}}$ on $Q$ and $D^*$, respectively, are
$\alpha$-related, that is,}
\begin{equation}\label{alpha-related}
(T_{q}\alpha)({\mathcal H}_{h,\alpha}^{\Lambda_{D^*}}(q)) =
{\mathcal H}_{h}^{\Lambda_{D^*}}(\alpha(q)), \; \; \;  for \; all
\; \; q\in Q.
\end{equation}
Therefore, we must prove that
\begin{center}
(i') $\Longleftrightarrow$ (ii)
\end{center}

(i') $\Longrightarrow$ (ii) Let $q$ be a point of $Q$. Then, using
(\ref{HhFolcharac}), (\ref{DefLalD}) and (\ref{alpha-related}), we
deduce that
\[
{\mathcal H}_{h}^{\Lambda_{D^*}}(\alpha(q)) \in {\mathcal
L}_{\alpha, D}(q).
\]
Consequently, from Proposition \ref{primerresultado}, we obtain
that
\[
{\mathcal H}_{h}^{\Lambda_{D^*}}(\alpha(q)) =
\#_{\Lambda_{D^*}}(\eta_{\alpha(q)}), \; \; \mbox{ for some }
\eta_{\alpha(q)} \in ({\mathcal L}_{\alpha, D}(q))^{0}.
\]
Thus, since ${\mathcal H}_{h}^{\Lambda_{D^*}}(\alpha(q)) =
-\#_{\Lambda_{D^*}}(dh(\alpha(q)))$, it follows that
\[
\eta_{\alpha(q)} + dh(\alpha(q)) \in Ker
(\#_{\Lambda_{D^*}}(\alpha(q))).
\]
Now, using Proposition \ref{segundoresultado} and the fact that
$\eta_{\alpha(q)} \in ({\mathcal L}_{\alpha, D}(q))^0$, we
conclude that
\begin{equation}\label{keypoint}
dh(\alpha(q)) \in ({\mathcal L}_{\alpha, D}(q))^0.
\end{equation}
Finally, if $a_{q} \in D_{q}$, we have that
\[
d^D(h \circ \alpha)(q)(a_{q}) =
dh(\alpha(q))((T_{q}\alpha)(\rho_{D}(a_{q})))
\]
which implies that (see (\ref{DefLalD}) and (\ref{keypoint}))
\[
d^D(h \circ \alpha)(q)(a_{q}) = 0.
\]

(ii) $\Longrightarrow$ (i') Let $q$ be a point of $Q$. Then, using
that $d^D(h\circ \alpha)(q) = 0$, we deduce that
\[
dh(\alpha(q)) \in ({\mathcal L}_{\alpha, D}(q))^0.
\]
Therefore, it follows that
\[
{\mathcal H}_{h}^{\Lambda_{D^*}}(\alpha(q)) =
-\#_{\Lambda_{D^*}}(dh(\alpha(q))) \in
\#_{\Lambda_{D^*}}(({\mathcal L}_{\alpha, D}(q))^0)
\]
and, from Proposition \ref{primerresultado}, we obtain that there
exists $v_{q} \in \rho_{D}(D_{q}) \subseteq T_{q}Q$ such that
\[
{\mathcal H}_{h}^{\Lambda_{D^*}}(\alpha(q)) =
(T_{q}\alpha)(v_{q}).
\]
This implies that
\[
{\mathcal H}_{h,\alpha}^{\Lambda_{D^*}}(q) =
(T_{\alpha(q)}\tau_{D^*})({\mathcal
H}_{h}^{\Lambda_{D^*}}(\alpha(q)) = v_{q}
\]
and, thus,
\[
{\mathcal H}_{h}^{\Lambda_{D^*}}(\alpha(q)) =
(T_{q}\alpha)({\mathcal H}_{h,\alpha}^{\Lambda_{D^*}}(q)).
\]
\hfill$\Box$

Let $(D, \{\cdot, \cdot\}_{D^*}, h)$ be a Hamiltonian system and
$\alpha: Q \to D^*$ be a section of the vector bundle $\tau_{D^*}:
D^* \to Q$.

Suppose that $(q^i)$ are local coordinates on an open subset $U$
of $Q$ and that $\{X_{\gamma}\}$ is a basis of sections of the
vector bundle $\tau_{D}^{-1}(U) \to U$. Denote by $(q^i,
p_{\gamma})$ the corresponding local coordinates on $D^*$ and by
$\rho^i_{\gamma}$, $C^{\delta}_{\gamma\nu}$ the local structure
functions of the skew-symmetric algebroid structure $(\lcf \cdot,
\cdot\rcf_{D}, \rho_{D})$ with respect to the local coordinates
$(q^i)$ and to the basis $\{X_{\gamma}\}$. If the local expression
of $\alpha$ is
\[
\alpha(q^i) = (q^i, \alpha_{\gamma}(q^i))
\]
then
\[
d^D\alpha = 0 \Longleftrightarrow
C_{\gamma\nu}^{\delta}\alpha_{\delta} = \rho^{i}_{\gamma}
\displaystyle \frac{\partial \alpha_{\nu}}{\partial q^i} -
\rho^i_{\nu} \frac{\partial \alpha_{\gamma}}{\partial q^i}, \; \;
\forall \gamma, \nu,
\]
and
\[
d^D(h \circ \alpha) = 0 \Longleftrightarrow \rho^i_{\gamma}(q)
(\displaystyle \frac{\partial h}{\partial q^i}_{|\alpha(q)} +
\frac{\partial \alpha_{\nu}}{\partial q^i}_{|q} \frac{\partial
h}{\partial p_{\nu}}_{|\alpha(q)}) = 0, \; \; \forall \gamma, \;
\forall q \in U.
\]

\begin{corollary}\label{strong-HJ}
Under the same hypotheses as in Theorem \ref{maintheorem} if,
additionally, $H^0(d^D) \simeq \R$ or if the skew-symmetric
algebroid $(D, \lcf \cdot , \rcf_{D}, \rho_{D})$ is completely
nonholonomic and $Q$ is connected, then the following conditions
are equivalent:
\begin{enumerate}
\item
If $c: I \to Q$ is an integral curve of the vector field
${\mathcal H}_{h,\alpha}^{\Lambda_{D^*}}$, that is,
\[
\dot{c}(t) = (T_{\alpha(c(t))}\tau_{D^*})({\mathcal
H}_{h}^{\Lambda_{D^*}}(\alpha(c(t)))), \; \; \mbox{ for all } t\in
I,
\]
then $\alpha \circ c: I \to D^*$ is a solution of the Hamilton
equations for $h$.

\item
$\alpha$ satisfies the following relation
\[
h \circ \alpha = \mbox{ constant }.
\]
\end{enumerate}
\end{corollary}
Note that if $Q$ is connected then
\[
h\circ \alpha = \mbox{ constant } \Longleftrightarrow
\displaystyle (\frac{\partial h}{\partial q^i}_{|\alpha(q)} +
\frac{\partial \alpha_{\nu}}{\partial q^i}_{|q} \frac{\partial
h}{\partial p_{\nu}}_{|\alpha(q)}) = 0, \; \; \forall i \mbox{ and
} \forall q \in U.
\]

\subsection{Linear almost Poisson morphisms and Hamilton-Jacobi
equation}

Suppose that $\tau_{D}: D \to Q$ and $\tau_{\bar{D}}: \bar{D} \to
\bar{Q}$ are vector bundles over $Q$ and $\bar{Q}$, respectively,
and that $\{\cdot, \cdot\}_{D^*}$ (respectively, $\{\cdot,
\cdot\}_{\bar{D}^*}$) is a linear almost Poisson structure on
$D^*$ (respectively, $\bar{D}^*$). Denote by $(\lcf\cdot,
\cdot\rcf_{D}, \rho_{D})$ and $d^D$ (respectively, $(\lcf\cdot,
\cdot\rcf_{\bar{D}}, \rho_{\bar{D}})$ and $d^{\bar{D}}$) the
corresponding skew-symmetric algebroid structure and the almost
differential on the vector bundle $\tau_{D}: D \to Q$
(respectively, $\tau_{\bar{D}}: \bar{D} \to \bar{Q}$).

\begin{definition}
A vector bundle morphism $(\tilde{F}, F)$ between the vector
bundles $\tau_{D^*}: D^* \to Q$ and $\tau_{\bar{D}^*}: \bar{D}^*
\to \bar{Q}$

\begin{picture}(375,90)(40,5)
\put(195,20){\makebox(0,0){$Q$}}
\put(250,25){$F$}\put(210,20){\vector(1,0){80}}
\put(305,20){\makebox(0,0){$\bar{Q}$}} \put(175,50){$\tau_{D^*}$}
\put(195,70){\vector(0,-1){40}} \put(310,50){$\tau_{\bar{D}^*}$}
\put(305,70){\vector(0,-1){40}} \put(195,80){\makebox(0,0){$D^*$}}
\put(250,85){$\tilde{F}$}\put(210,80){\vector(1,0){80}}
\put(305,80){\makebox(0,0){$\bar{D}^*$}}
\end{picture}

\noindent is said to be a \emph{linear almost Poisson morphism} if
\begin{equation}\label{lalmorph}
\{\bar{\varphi} \circ \tilde{F}, \bar{\psi}\circ \tilde{F}\}_{D^*}
= \{\bar{\varphi}, \bar{\psi}\}_{\bar{D}^*} \circ \tilde{F},
\end{equation}
for $\bar{\varphi}, \bar{\psi} \in C^{\infty}(\bar{D}^*)$.
\end{definition}
Let $(\tilde{F}, F)$ be a vector bundle morphism between the
vector bundles $\tau_{D^*}: D^* \to Q$ and $\tau_{\bar{D}^*}:
\bar{D}^* \to \bar{Q}$. If $\bar{X}$ is a section of
$\tau_{\bar{D}}: \bar{D} \to \bar{Q}$ then we may define the
section $(\tilde{F}, F)^*\bar{X}$ of $\tau_{D}: D \to Q$
characterized by the following condition
\begin{equation}\label{Pullback}
\alpha_{q}(((\tilde{F}, F)^*\bar{X})(q)) =
\tilde{F}(\alpha_{q})(\bar{X}(F(q))),
\end{equation}
for all $q\in Q$ and $\alpha_{q} \in D^*_{q}$.

\begin{theorem}\label{maintheorem2}
Let $(\tilde{F}, F)$ be a vector bundle morphism between the
vector bundles $\tau_{D^*}: D^* \to Q$ and $\tau_{\bar{D}^*}:
\bar{D}^* \to \bar{Q}$. Then, $(\tilde{F}, F)$ is a linear almost
Poisson morphism if and only if
\begin{equation}\label{Firstcond}
\lcf(\tilde{F}, F)^*\bar{X}, (\tilde{F}, F)^*\bar{Y}\rcf_{D} =
(\tilde{F}, F)^*\lcf\bar{X}, \bar{Y}\rcf_{\bar{D}},
\end{equation}
\vspace{-.5cm}
\begin{equation}\label{Secondcond}
(TF \circ \rho_{D})((\tilde{F}, F)^*\bar{X}) =
\rho_{\bar{D}}(\bar{X})\circ F,
\end{equation}
for $\bar{X}, \bar{Y} \in \Gamma(\tau_{\bar{D}})$.
\end{theorem}
\begin{proof}
Suppose that $(\tilde{F}, F)$ is a linear almost Poisson morphism
and that $\bar{Z}$ is a section of $\tau_{\bar{D}}: \bar{D} \to
\bar{Q}$. From (\ref{Pullback}), it follows that
\begin{equation}\label{Equtil}
\lhat{40}{(\tilde{F}, F)^*\bar{Z}} = \hat{\bar{Z}} \circ
\tilde{F}.
\end{equation}
Now, if $\bar{X}, \bar{Y} \in \Gamma(\tau_{D})$ then, using
(\ref{alafromalp}) and (\ref{Equtil}), we deduce that
\[
\lhat{95}{\lcf(\tilde{F}, F)^*\bar{X}, (\tilde{F},
F)^*\bar{Y}\rcf_{D}} = -\{\hat{\bar{X}} \circ \tilde{F},
\hat{\bar{Y}} \circ \tilde{F}\}_{D^*}.
\]
Thus, from (\ref{alafromalp}) and (\ref{lalmorph}), we obtain that
\[
\lhat{95}{\lcf(\tilde{F}, F)^*\bar{X}, (\tilde{F},
F)^*\bar{Y}\rcf_{D}} = \lhat{60}{(\tilde{F}, F)^*\lcf \bar{X},
\bar{Y}\rcf_{\bar{D}}}
\]
which implies that (\ref{Firstcond}) holds.

On the other hand, if $\bar{f} \in C^{\infty}(\bar{Q})$ then,
using again (\ref{alafromalp}) and (\ref{lalmorph}), it follows
that
\[
(\rho_{\bar{D}}(\bar{X})(\bar{f})\circ F)\circ \tau_{D^*} =
\{(\bar{f} \circ \tau_{\bar{D}^*})\circ \tilde{F}, \hat{\bar{X}}
\circ \tilde{F}\}_{D^*}.
\]
Therefore, from (\ref{alafromalp}) and (\ref{Equtil}), we have
that
\[
(\rho_{\bar{D}}(\bar{X})(\bar{f})\circ F)\circ \tau_{D^*} =
\rho_{D}((\tilde{F}, F)^*\bar{X})(\bar{f}\circ F)\circ \tau_{D^*}
\]
and, consequently,
\[
\rho_{\bar{D}}(\bar{X})(\bar{f})\circ F= \rho_{D}((\tilde{F},
F)^*\bar{X})(\bar{f}\circ F).
\]
This implies that (\ref{Secondcond}) holds.

Conversely, assume that (\ref{Firstcond}) and (\ref{Secondcond})
hold.

Then, if $\bar{f}, \bar{g} \in C^{\infty}(\bar{Q})$ it is clear
that the real functions
\[
\bar{f} \circ \tau_{\bar{D}^*} \circ \tilde{F} = \bar{f} \circ F
\circ \tau_{D^*}, \; \; \; \bar{g} \circ \tau_{\bar{D}^*} \circ
\tilde{F} = \bar{g} \circ F \circ \tau_{D^*}
\]
are basic functions with respect to the projection $\tau_{D^*}:
D^* \to Q$. Therefore, we deduce that
\begin{equation}\label{iii1}
0 = \{(\bar{f} \circ \tau_{\bar{D}^*})\circ \tilde{F}, (\bar{g}
\circ \tau_{\bar{D}^*})\circ \tilde{F}\}_{D^*} = \{\bar{f} \circ
\tau_{\bar{D}^*}, \bar{g} \circ \tau_{\bar{D}^*}\}_{\bar{D}^*}
\circ \tilde{F}.
\end{equation}
Now, if $\bar{X} \in \Gamma(\tau_{\bar{D}})$ then, using
(\ref{alafromalp}) and (\ref{Equtil}), we obtain that
\[
\{(\bar{f} \circ \tau_{\bar{D}^*})\circ \tilde{F},
\hat{\bar{X}}\circ \tilde{F}\}_{D^*} = ((TF \circ
\rho_{D})((\tilde{F}, F)^*\bar{X}))(\bar{f})\circ \tau_{D^*}.
\]
Consequently, from (\ref{alafromalp}) and (\ref{Secondcond}), it
follows that
\begin{equation}\label{iii2}
\{(\bar{f} \circ \tau_{\bar{D}^*})\circ \tilde{F},
\hat{\bar{X}}\circ \tilde{F}\}_{D^*} = \{\bar{f}\circ
\tau_{\bar{D}^*}, \hat{\bar{X}}\}_{\bar{D}^*} \circ \tilde{F}.
\end{equation}
On the other hand, if $\bar{Y} \in \Gamma(\tau_{\bar{D}})$ then,
using (\ref{alafromalp}), (\ref{Firstcond}) and (\ref{Equtil}), we
deduce that
\begin{equation}\label{iii3}
\{\hat{\bar{X}}\circ \tilde{F}, \hat{\bar{Y}}\circ
\tilde{F}\}_{D^*} = \{\hat{\bar{X}}, \hat{\bar{Y}}\}_{\bar{D}^*}
\circ \tilde{F}.
\end{equation}
Thus, (\ref{iii1}), (\ref{iii2}) and (\ref{iii3}) imply that
$(\tilde{F}, F)$ is a linear almost Poisson morphism.
\end{proof}

Let $(\tilde{F}, F)$ be a vector bundle morphism between the
vector bundles $\tau_{D^*}: D^* \to Q$ and $\tau_{\bar{D}^*}:
\bar{D}^* \to \bar{Q}$. Denote by $\Lambda^k\tilde{F}:
\Lambda^kD^*\to \Lambda^k\bar{D}^*$ the vector bundle morphism
(over $F$) between the vector bundles $\Lambda^k\tau_{D^*}:
\Lambda^kD^* \to Q$ and $\Lambda^k\tau_{\bar{D}^*}:
\Lambda^k\bar{D}^* \to \bar{Q}$ induced by $\tilde{F}$. Then, a
section $\alpha \in \Gamma(\Lambda^k\tau_{D^*})$ is said to be
$(\tilde{F}, F)$-related with a section $\bar{\alpha} \in
\Gamma(\Lambda^k\tau_{\bar{D}^*})$ if
\[
\Lambda^k \tilde{F} \circ \alpha = \bar{\alpha} \circ F.
\]
Now, assume that $F$ is a surjective map and that $(\tilde{F}, F)$
is a fiberwise injective vector bundle morphism, that is,
\[
\tilde{F}_{q} = \tilde{F}_{|D^*_q}: D^*_q \to \bar{D}^*_{F(q)}
\]
is a monomorphism of vector spaces, for all $q\in Q$, and
\[
F(q) = F(q') \Longrightarrow \tilde{F}_{q}(D^*_{q}) =
\tilde{F}_{q'}(D^*_{q'}).
\]
Then, we may consider the vector subbundle $\tilde{F}(D^*)$ (over
$\bar{Q}$) of $\tau_{\bar{D}^*}: \bar{D}^* \to \bar{Q}$. Moreover,
if $\bar{\alpha}$ is a section of this vector subbundle we have
that there exists a unique section $\alpha$ of $\tau_{D^*}: D^*
\to Q$ such that $\alpha$ is $(\tilde{F}, F)$-related with
$\bar{\alpha}$. In fact, if $\{\bar{\alpha}_{i}\}$ is a local
basis of sections of $\tau_{\tilde{F}(D^*)}: \tilde{F}(D^*) \to
\bar{Q}$, it follows that $\{\alpha_{i}\}$ is a local basis of
$\Gamma(\tau_{D^*})$.

\begin{theorem}\label{mordiff}
Let $(\tilde{F}, F)$ be a vector bundle morphism between the
vector bundles $\tau_{D^*}: D^* \to Q$ and $\tau_{\bar{D}^*}:
\bar{D}^* \to \bar{Q}$.
\begin{enumerate}
\item
If $(\tilde{F}, F)$ is a linear almost Poisson morphism then the
following condition ({\bf C}) holds:

({\bf C}) For each $\alpha \in \Gamma(\Lambda^k\tau_{D^*})$ which
is $(\tilde{F}, F)$-related with $\bar{\alpha} \in
\Gamma(\Lambda^k\tau_{\bar{D}^*})$ we have that $d^D\alpha \in
\Gamma(\Lambda^{k+1}\tau_{D^*})$ is also $(\tilde{F}, F)$-related
with $d^{\bar{D}}\bar{\alpha} \in
\Gamma(\Lambda^{k+1}\tau_{\bar{D}^*})$.
\item
Conversely, if condition ({\bf C}) holds, $F$ is a surjective map
and $(\tilde{F}, F)$ is a fiberwise injective vector bundle
morphism then $(\tilde{F}, F)$ is a linear almost Poisson
morphism.
\end{enumerate}
\end{theorem}
\begin{proof}
(i) Suppose that $\bar{X}$ and $\bar{Y}$ are sections of
$\tau_{\bar{D}}: \bar{D} \to \bar{Q}$.

Then, if $\bar{f} \in C^{\infty}(\bar{Q})$, using
(\ref{Secondcond}), we deduce that
\[
(d^{\bar{D}}\bar{f})(\bar{X})\circ F = (\rho_{D}((\tilde{F},
F)^*\bar{X}))(\bar{f}\circ F).
\]
Thus, from (\ref{Pullback}), it follows that
\[
(d^{\bar{D}}\bar{f})(\bar{X})\circ F = <\tilde{F}(d^D(\bar{f}\circ
F)), \bar{X} \circ F>.
\]
Therefore, we have that
\begin{equation}\label{Tres1}
d^{\bar{D}}\bar{f} \circ F = \tilde{F} \circ d^D(\bar{f} \circ F).
\end{equation}
Now, let $\alpha$ be a section of $\tau_{D^*}: D^* \to Q$ which is
$(\tilde{F}, F)$-related with $\bar{\alpha} \in
\Gamma(\tau_{\bar{D}^*})$, that is,
\begin{equation}\label{Ftilderelated}
\tilde{F} \circ \alpha = \bar{\alpha}\circ F.
\end{equation}
Then, using (\ref{Firstcond}), (\ref{Secondcond}) and
(\ref{Ftilderelated}), we obtain that
\begin{eqnarray*}
(d^{\bar{D}}\bar{\alpha})(\bar{X}, \bar{Y})\circ F &=&
\rho_{D}((\tilde{F}, F)^*\bar{X})(\alpha((\tilde{F},
F)^*\bar{Y}))- \rho_{D}((\tilde{F},
F)^*\bar{Y})(\alpha((\tilde{F}, F)^*\bar{X}))\\ &&-
\alpha\lcf(\tilde{F},F)^*\bar{X}, (\tilde{F}, F)^*\bar{Y}\rcf_{D}
\end{eqnarray*}
which implies that
\[
(d^{\bar{D}}\bar{\alpha})(\bar{X}, \bar{Y})\circ F =
<\Lambda^2\tilde{F}\circ d^D\alpha, (\bar{X} \circ F, \bar{Y}
\circ F)>.
\]
This proves that
\begin{equation}\label{Tres2}
d^{\bar{D}}\bar{\alpha} \circ F = \Lambda^2\tilde{F} \circ
d^D\alpha.
\end{equation}
Consequently, from (\ref{Derivation}), (\ref{Tres1}) and
(\ref{Tres2}), we deduce the result.

(ii) If $\bar{X} \in \Gamma(\tau_{\bar{D}})$ and $\bar{f}\in
C^{\infty}(\bar{Q})$ then, using condition ({\bf C}), we have that
\[
(\rho_{\bar{D}}(\bar{X})\circ F)(\bar{f}) = d^D(\bar{f}\circ
F)((\tilde{F}, F)^*\bar{X}) = (\rho_{D}((\tilde{F},
F)^*\bar{X}))(\bar{f}\circ F).
\]
This proves that (\ref{Secondcond}) holds.

Next, suppose that $\bar{Y} \in \Gamma(\tau_{\bar{D}})$ and that
$\alpha$ is a section of $\tau_{D^*}: D^* \to Q$ which is
$(\tilde{F}, F)$-related with $\bar{\alpha} \in
\Gamma(\tau_{\bar{D}^*})$.

Then, from (\ref{Secondcond}), it follows that
\[
\alpha\lcf(\tilde{F}, F)^*\bar{X}, (\tilde{F}, F)^*\bar{Y}\rcf_{D}
= -(d^D\alpha)((\tilde{F}, F)^*\bar{X}, (\tilde{F}, F)^*\bar{Y}) +
\rho_{\bar{D}}(\bar{X})(\bar{\alpha}(\bar{Y}))\circ F -
\rho_{\bar{D}}(\bar{Y})(\bar{\alpha}(\bar{X}))\circ F.
\]
Thus, using condition ({\bf C}), we deduce that
\[
\alpha\lcf(\tilde{F}, F)^*\bar{X}, (\tilde{F}, F)^*\bar{Y}\rcf_{D}
= -(d^{\bar D}\bar{\alpha})(\bar{X},\bar{Y}) \circ F +
\rho_{\bar{D}}(\bar{X})(\bar{\alpha}(\bar{Y}))\circ F -
\rho_{\bar{D}}(\bar{Y})(\bar{\alpha}(\bar{X}))\circ F.
\]
This implies that (\ref{Firstcond}) holds.
\end{proof}
\begin{remark}{\rm Let $(\tilde{F}, F)$ be a linear almost Poisson
morphism between the vector bundles $\tau_{D^*}: D^* \to Q$ and
$\tau_{\bar{D}^*}: \bar{D}^* \to \bar{Q}$. Moreover, suppose that
$F$ is surjective and that $(\tilde{F}, F)$ is a fiberwise
injective vector bundle morphism.
\begin{enumerate}
\item
{}From Theorem \ref{mordiff}, we deduce that the condition
$H^{0}(d^D) \simeq \R$ implies that $H^{0}(d^{\bar{D}}) \simeq
\R$. In general, the converse does not hold. However, if
$H^0(d^{\bar{D}}) \simeq \R$ and $f \in C^{\infty}(Q)$ is a
$F$-basic function such that $d^D f = 0$ then $f$ is constant.
\item
If $F$ is a surjective submersion with connected fibers, $V_{q}F
\subseteq \tilde{D}_{q} = \rho_{D}(D_{q})$, for all $q \in Q$, and
$d^D f = 0$ then $f$ is a $F$-basic function. Here, $VF$ is the
vertical bundle to $F$.
\end{enumerate} }
\end{remark}

Now, we will introduce the following definition.
\begin{definition}
Let $(D, \{\cdot, \cdot\}_{D^*}, h)$ (respectively, $(\bar{D},
\{\cdot, \cdot\}_{\bar{D}^*}, \bar{h})$) be a Hamiltonian system
and $(\tilde{F}, F)$ be a linear almost Poisson morphism between
the vector bundles $\tau_{D}: D^* \to Q$ and $\tau_{\bar{D}}:
\bar{D}^* \to \bar{Q}$. Then, $(\tilde{F}, F)$ is said to be
\emph{Hamiltonian} if
\[
\bar{h} \circ \tilde{F} = h.
\]
\end{definition}
It is clear that if $(\tilde{F}, F)$ is a Hamiltonian morphism
then the Hamiltonian vector fields of $h$ and $\bar{h}$,
${\mathcal H}^{\Lambda_{D^*}}_{h}$ and ${\mathcal
H}_{\bar{h}}^{\Lambda_{\bar{D}^*}}$, are $\tilde{F}$-related, that
is,
\[
(T_{\beta}\tilde{F})({\mathcal H}_{h}^{\Lambda_{D^*}}(\beta)) =
{\mathcal H}_{\bar{h}}^{\Lambda_{\bar{D}^*}}(\tilde{F}(\beta)), \;
\; \; \mbox{ for all } \beta \in D^*.
\]
This implies that if $\mu: I \to D^*$ is a solution of the
Hamilton equations for $h$ then $\tilde{F} \circ \mu: I \to
\bar{D}^*$ is a solution of the Hamilton equations for $\bar{h}$.

In addition, from Theorem \ref{mordiff}, we deduce the following
result
\begin{theorem}\label{morHa-Ja}
Let $(D, \{\cdot, \cdot\}_{D^*}, h)$ (respectively, $(\bar{D},
\{\cdot, \cdot\}_{\bar{D}^*}, \bar{h})$) be a Hamiltonian system
and $(\tilde{F}, F)$ be a Hamiltonian morphism between the vector
bundles $\tau_{D^*}: D^{*} \to Q$ and $\tau_{\bar{D}^*}: \bar{D}^*
\to \bar{Q}$. Assume that the map $F$ is surjective and that
$(\tilde{F}, F)$ is a fiberwise injective vector bundle morphism.
\begin{enumerate}
\item
If $\alpha$ is a section of the vector bundle $\tau_{D^*}: D^* \to
Q$ such that $d^D\alpha = 0$, it satisfies the Hamilton-Jacobi
equation for $h$ (respectively, the strongest condition $h\circ
\alpha = \mbox{ constant }$) and it is $(\tilde{F}, F)$-related
with $\bar{\alpha} \in \Gamma(\tau_{\bar{D}^*})$ then
$d^{\bar{D}}\bar{\alpha} = 0$ and $\bar{\alpha}$ satisfies the
Hamilton-Jacobi equation for $\bar{h}$ (respectively, the
strongest condition $\bar{h} \circ \bar{\alpha} = \mbox{ constant
}$).
\item
If $\bar{\alpha}$ is a section of the vector subbundle
$\tilde{F}(D^*)$ of $\tau_{\bar{D}^*}: \bar{D}^* \to \bar{Q}$ such
that $d^{\bar{D}}\bar{\alpha} = 0$ and $\bar{\alpha}$ satisfies
the Hamilton-Jacobi equation for $\bar{h}$ (respectively, the
strongest condition $\bar{h} \circ \bar{\alpha} = \mbox{ constant
}$) then $d^D\alpha = 0$ and $\alpha$ satisfies the
Hamilton-Jacobi equation for $h$ (respectively, the strongest
condition $h \circ \alpha = \mbox{ constant }$), where $\alpha$ is
the section of $\tau_{D^*}: D^* \to Q$ characterized by the
condition $\tilde{F} \circ \alpha = \bar{\alpha} \circ F$.

\end{enumerate}

\end{theorem}

\section{Applications to nonholonomic Mechanics}\label{section4}

\subsection{Unconstrained mechanical systems on a Lie
algebroid}\label{unconmesys}

Let $\tau_{A}: A \to Q$ be a Lie algebroid over a manifold $Q$ and
denote by $(\lcf\cdot, \cdot\rcf_{A}, \rho_{A})$ the Lie algebroid
structure of $A$.

If ${\mathcal G}: A \times_{Q} A \to \R$ is a bundle metric on $A$
then the \emph{Levi-Civita connection}
\[
\nabla^{\mathcal G}: \Gamma(\tau_{A}) \times \Gamma(\tau_{A}) \to
\Gamma(\tau_{A})
\]
is determined by the formula
\[
\begin{array}{rcl}
2 {\mathcal G}(\nabla_{X}^{\mathcal G}Y, Z) & = &
\rho_{A}(X)({\mathcal G}(Y, Z)) + \rho_{A}(Y)({\mathcal G}(X, Z))
- \rho_{A}(Z)({\mathcal G}(X, Y)) \\
&& +  {\mathcal G}(X, \lcf Z, Y\rcf_{A}) + {\mathcal G}(Y, \lcf Z,
X\rcf_{A}) - {\mathcal G}(Z, \lcf Y, X\rcf_{A})
\end{array}
\]
for $X, Y, Z \in \Gamma(A)$. Using the covariant derivative
induced by $\nabla^{\mathcal G}$, one may introduce the notion of
a geodesic of $\nabla^{\mathcal G}$ as follows. A curve $\sigma: I
\to A$ is \emph{admissible} if
\[
\displaystyle \frac{d}{dt}(\tau_{A} \circ \sigma) = \rho_{A} \circ
\sigma.
\]
An admissible curve $\sigma: I \to A$ is said to be a
\emph{geodesic} if $\nabla_{\sigma(t)}^{\mathcal G}\sigma(t) = 0$,
for all $t \in I$.

The geodesics are the integral curves of a vector field
$\xi_{\mathcal G}$ on $A$, the \emph{geodesic flow} of $A$, which
is locally given by
\[
\xi_{\mathcal G} = \displaystyle \rho_{B}^{i} v^{B}
\frac{\partial}{\partial q^i}- C_{EB}^{C}v^{B}v^{C}
\frac{\partial}{\partial v^{E}}.
\]
Here, $(q^i)$ are local coordinates on an open subset $U$ of $Q$,
$\{X_{B}\}$ is an orthonormal basis of sections of the vector
bundle $\tau_{A}^{-1}(U) \to U$, $(q^i, v^{B})$ are the
corresponding local coordinates on $A$ and $\rho^{i}_{B}$,
$C^{E}_{CB}$ are the local structure functions of $A$. Note that
the coefficients $\Gamma_{BC}^{E}$ of the connection
$\nabla^{\mathcal G}$ are
\[
\Gamma_{BC}^{E} = \displaystyle \frac{1}{2} (C_{EB}^{C} +
C_{EC}^{B} + C_{BC}^{E})
\]
(for more details, see \cite{CoLeMaMa,CoMa}).

The \emph{Lagrangian function} $L: A \to \R$ of an (unconstrained)
mechanical system on $A$ is defined by
\[
L(a) = \displaystyle \frac{1}{2} {\mathcal G}(a, a) -
V(\tau_{A}(a)) = \frac{1}{2}\|a\|_{\mathcal G}^2 - V(\tau_{A}(a)),
\; \; \; \mbox{ for } a\in A,
\]
$V: Q \to \R$ being a real $C^{\infty}$-function on $Q$. In other
words, $L$ is the kinetic energy induced by ${\mathcal G}$ minus
the potential energy induced by $V$.

Note that if $\Delta$ is the Liouville vector field of $A$ then
\emph{the Lagrangian energy} $E_{L} = \Delta(L) - L$ is the real
$C^{\infty}$-function on $A$ given by
\[
E_{L}(a) = \displaystyle \frac{1}{2} {\mathcal G}(a, a) +
V(\tau_{A}(a)) = \frac{1}{2} \|a\|_{\mathcal G}^2 +
V(\tau_{A}(a)), \; \; \; \mbox{ for } a\in A.
\]
On the other hand, we may consider the section $grad_{\mathcal
G}V$ of $\tau_{A}: A \to Q$ characterized by the following
condition
\[
{\mathcal G}(grad_{\mathcal G}V, X) = (d^A V)(X) = \rho_{A}(X)(V),
\; \; \forall X \in \Gamma(\tau_{A}).
\]
Then, the solutions of \emph{the Euler-Lagrange equations} for $L$
are the integral curves of the vector field $\xi_{L}$ on $A$
defined by
\[
\xi_{L} = \xi_{\mathcal G} - (grad_{\mathcal G}V)^{\bf v},
\]
where $(grad_{\mathcal G}V)^{\bf v} \in {\frak X}(A)$ is the
standard vertical lift of the section $grad_{\mathcal G}V$. The
local expression of the Euler-Lagrange equations is
\[
\dot{q}^i = \rho_{B}^i v^{B}, \makebox[.4cm]{} \dot{v}^{E} =
-C_{EB}^{C}v^{B}v^{C} - \rho_{E}^{j} \frac{\partial V}{\partial
q^{j}},
\]
for all $i$ and $E$ (see \cite{CoLeMaMa,CoMa}).

Now, we will denote by $\flat_{\mathcal G}: A \to A^*$ the vector
bundle isomorphism induced by ${\mathcal G}$ and by $\#_{\mathcal
G}: A^* \to A$ the inverse morphism. If $\alpha: Q \to A^*$ is a
section of the vector bundle $\tau_{A^*}: A^* \to Q$ we also
consider the vector field $\xi_{L,\alpha}$ on $Q$ defined by
\[
\xi_{L,\alpha}(q) = (T_{\#_{\mathcal
G}(\alpha(q))}\tau_{A})(\xi_{L}(\#_{\mathcal G}(\alpha (q)))), \;
\; \mbox{ for } q\in Q.
\]

\begin{corollary}\label{unconsHam-Ja}
Let $\alpha: Q \to A^*$ be a $1$-cocycle of the Lie algebroid $A$,
that is, $d^A\alpha = 0$. Then, the following conditions are
equivalent:
\begin{enumerate}
\item
If $c: I \to Q$ is an integral curve of the vector field
$\xi_{L,\alpha}$ on $Q$ we have that $\#_{\mathcal G} \circ \alpha
\circ c: I \to A$ is a solution of the Euler-Lagrange equations
for $L$.

\item
$\alpha$ satisfies the \emph{Hamilton-Jacobi equation}
\[
d^A(E_{L} \circ \#_{\mathcal G} \circ \alpha) = 0,
\]
that is, the function $\displaystyle \frac{1}{2} \|\#_{\mathcal
G}\circ \alpha\|_{\mathcal G}^2 + V$ on $Q$ is constant on the
leaves of the Lie algebroid foliation associated with $A$.
\end{enumerate}

\end{corollary}

\begin{proof}
The Legendre transformation associated with the Lagrangian
function $L$ is the vector bundle isomorphism $\flat_{\mathcal G}:
A \to A^*$ between $A$ and $A^*$ induced by the bundle metric
${\mathcal G}$ (for the definition of the Legendre transformation
associated with a Lagrangian function on a Lie algebroid, see
\cite{LeMaMa}). Thus, if we denote by ${\mathcal G}^*$ the bundle
metric on $A^*$ then, the Hamiltonian function $H_{L} = E_{L}
\circ \#_{\mathcal G}$ induced by the hyperregular Lagrangian
function $L$ is given by
\[
H_{L}(\gamma) = \displaystyle \frac{1}{2} {\mathcal G}^*(\gamma,
\gamma) + V(\tau_{A^*}(\gamma)), \; \; \; \mbox{ for } \gamma \in
A^*.
\]
Therefore, if $\Lambda_{A^*}$ is the corresponding linear Poisson
$2$-vector on $A^*$ and ${\mathcal H}_{H_L}^{\Lambda_{A^*}}$ is
the Hamiltonian vector field of $H_L$ with respect to
$\Lambda_{A^*}$, we have that the solutions of the Hamilton
equations are the integral curves of the vector field ${\mathcal
H}_{H_L}^{\Lambda_{A^*}}$. In fact, the vector fields $\xi_{L}$
and ${\mathcal H}_{H_L}^{\Lambda_{A^*}}$ are $\flat_{\mathcal
G}$-related, that is,
\[
T\flat_{\mathcal G} \circ \xi_{L} = {\mathcal
H}_{H_L}^{\Lambda_{A^*}} \circ \flat_{\mathcal G}.
\]
Consequently, if $\sigma: I \to A$ is a solution of the
Euler-Lagrange equations for $L$ then $\flat_{\mathcal G} \circ
\sigma: I \to A^*$ is a solution of the Hamilton equations for
$H_L$ and, conversely, if $\gamma: I \to A^*$ is a solution of the
Hamilton equations for $H_{L}$ then $\#_{\mathcal G} \circ \gamma:
I \to A$ is a solution of the Euler-Lagrange equations for $L$
(for more details, see \cite{LeMaMa}).

In addition, since $\tau_{A^*} \circ \flat_{\mathcal G} =
\tau_{A}$, it follows that
\[
\xi_{L,\alpha}(q) = (T_{\alpha(q)}\tau_{A^*})({\mathcal
H}_{H_L}^{\Lambda_{A^*}}(\alpha(q))) = {\mathcal
H}_{H_L,\alpha}^{\Lambda_{A^*}}(q), \; \; \mbox{ for } q \in Q,
\]
i.e., $\xi_{L,\alpha} = {\mathcal H}_{H_L,\alpha}^{\Lambda_{A^*}}$.

Thus, using Theorem \ref{maintheorem} (or, alternatively, using
Theorem 3.16 in \cite{LeMaMa}), we deduce the result.
\end{proof}
Next, we will apply Corollary \ref{unconsHam-Ja} to the particular
case when $A$ is the standard Lie algebroid $TQ$ and $\alpha$ is a
$1$-coboundary, that is, $\alpha = dS$ with $S: Q \to \R$ a real
$C^{\infty}$-function on $Q$. Note that, in this case, the bundle
metric ${\mathcal G}$ on $TQ$ is a Riemannian metric $g$ on $Q$
and that $\#_{\mathcal G} \circ \alpha = \#_{g} \circ dS$ is just
the gradient vector field of $S$, $grad_{g}S$, with respect to
$g$.
\begin{corollary}\label{Ham-Jastan}
Let $S: Q \to \R$ be a real $C^{\infty}$-function on $Q$. Then,
the following conditions are equivalent:
\begin{enumerate}
\item
If $c: I \to Q$ is an integral curve of the vector field
$\xi_{L,dS}$ on $Q$ we have that $grad_{g}S \circ c: I \to A$ is a
solution of the Euler-Lagrange equations for $L$.

\item
$S$ satisfies the \emph{Hamilton-Jacobi equation}
\[
d(E_{L} \circ grad_{g}S) = 0,
\]
that is, the function $\displaystyle \frac{1}{2} \| grad_{g}S
\|_{g}^2 + V$ on $Q$ is constant.
\end{enumerate}

\end{corollary}

\begin{remark}{\rm
Corollary \ref{Ham-Jastan} is a consequence of a well-known result
(see Theorem 5.2.4 in \cite{AbMa}).}
\end{remark}

Now, let $L: A \to \R$ (respectively, $\bar{L}: \bar{A}\to \R$) be
the Lagrangian function of an unconstrained mechanical system on a
Lie algebroid $\tau_{A}: A \to Q$ (respectively, $\tau_{\bar{A}}:
\bar{A} \to \bar{Q}$) and $(\tilde{F}, F)$ be a linear Poisson
morphism between the Poisson manifolds $(A^*, \{\cdot,
\cdot\}_{A^*})$ and $(\bar{A}^*, \{\cdot, \cdot\}_{\bar{A}^*})$
such that:
\begin{enumerate}
\item
$F: Q \to \bar{Q}$ is a surjective map.
\item
For each $q\in Q$, the linear map $\tilde{F}_{q} =
\tilde{F}_{|A^*_q}: A^*_q \to \bar{A}^*_{F(q)}$ satisfies the
following conditions
\[
\bar{\mathcal G}^*(\tilde{F}_{q}(\beta), \tilde{F}_{q}(\beta')) =
{\mathcal G}^*(\beta, \beta'), \; \; \; \mbox{ for } \beta, \beta'
\in A^*_q,
\]
\vspace{-.5cm}
\[
F(q) = F(q') \Longrightarrow \tilde{F}_q(A^*_q) =
\tilde{F}_{q'}(A^*_{q'}),
\]
where ${\mathcal G}^*$ (respectively, $\bar{\mathcal G}^*$) is the
bundle metric on $A^*$ (respectively, $\bar{A}^*$). Note that the
first condition implies that $\tilde{F}_{q}$ is injective and an isometry.
\item
If $V: Q \to \R$ (respectively, $\bar{V}: \bar{Q} \to \R$) is the
potential energy of the mechanical system on $A$ (respectively,
$\bar{A}$) we have that $\bar{V} \circ F = V$.
\end{enumerate}
Then, we deduce that $(\tilde{F}, F)$ is a Hamiltonian morphism
between the Hamiltonian systems $(A, \{\cdot, \cdot\}_{A^*},
\linebreak H_{L})$ and $(\bar{A}, \{\cdot, \cdot\}_{\bar{A}^*},
H_{\bar{L}})$, where $H_{L}$ (respectively, $H_{\bar{L}}$) is the
Hamiltonian function on $A^*$ (respectively, $\bar{A}^*$)
associated with the Lagrangian function $L$ (respectively,
$\bar{L}$).

Moreover, using Theorem \ref{morHa-Ja}, we conclude that
\begin{corollary}\label{relHam-Jaun}
\begin{enumerate}
\item
If $\alpha: Q \to A^*$ is a $1$-cocycle for the Lie algebroid $A$
($d^A\alpha = 0$), it satisfies the Hamilton-Jacobi equation
\begin{equation}\label{EqHa-Ja1}
d^A(E_{L} \circ \#_{\mathcal G} \circ \alpha) = 0
\end{equation}
and it is $(\tilde{F}, F)$-related with $\bar{\alpha} \in
\Gamma(\tau_{\bar{A}^*})$ then $d^{\bar{A}}\bar{\alpha} = 0$ and
$\bar{\alpha}$ is a solution of the Hamilton-Jacobi equation
\begin{equation}\label{EqHa-Ja2}
d^{\bar{A}}(E_{\bar{L}} \circ \#_{\bar{\mathcal G}} \circ
\bar{\alpha}) = 0.
\end{equation}
\item
If $\bar{\alpha}: \bar{Q} \to \tilde{F}(A^*) \subseteq \bar{A}^*$
is a $1$-cocycle for the Lie algebroid $\bar{A}$
($d^{\bar{A}}\bar{\alpha} = 0$) and it satisfies the
Hamilton-Jacobi equation (\ref{EqHa-Ja2}) then $d^{A}\alpha = 0$
and $\alpha$ is a solution of the Hamilton Jacobi equation
(\ref{EqHa-Ja1}). Here, $\alpha: Q \to A^*$ is the section of
$\tau_{A^*}: A^* \to Q$ characterized by the condition $\tilde{F}
\circ \alpha = \bar{\alpha} \circ F$.

\end{enumerate}
\end{corollary}
A particular example of the above general construction is the
following one.

Let $F: Q \to \bar{Q} = Q/G$ be a principal $G$-bundle. Denote by
$\phi: G \times Q \to Q$ the free action of $G$ on $Q$ and by
$T\phi: G \times TQ \to TQ$ the tangent lift of $\phi$. $T\phi$ is
a free action of $G$ on $TQ$. Then, we may consider the quotient
vector bundle $\tau_{\bar{A}}= \tau_{TQ/G}: \bar{A} = TQ/G \to
\bar{Q} = Q/G$. The sections of this vector bundle may be
identified with the vector fields on $Q$ which are $G$-invariant.
Thus, using that a $G$-invariant vector field is $F$-projectable
and that the standard Lie bracket of two $G$-invariant vector
fields is also a $G$-invariant vector field, we can define a Lie
algebroid structure $(\lcf\cdot, \cdot\rcf_{\bar{A}},
\rho_{\bar{A}})$ on the quotient vector bundle $\tau_{\bar{A}}=
\tau_{TQ/G}: \bar{A} = TQ/G \to \bar{Q} = Q/G$. The resultant Lie
algebroid is called the \emph{Atiyah (gauge) algebroid associated
with the principal bundle $F: Q \to \bar{Q} = Q/G$} (see
\cite{LeMaMa,Mac}).

On the other hand, denote by $T^*\phi: G \times T^*Q \to T^*Q$ the
cotangent lift of the action $\phi$. Then, the space of orbits of
$T^*\phi$, $T^*Q/G$, may be identified with the dual bundle
$\bar{A}^*$ to $\bar{A}$. Under this identification, the linear
Poisson structure on $\bar{A}^*$ is characterized by the following
condition: the canonical projection $\tilde{F}: A^* = T^*Q \to
T^*Q/G \simeq \bar{A}^*$ is a Poisson morphism, when on $A^* =
T^*Q$ we consider the linear Poisson structure induced by the
standard Lie algebroid $\tau_{A}= \tau_{TQ}: A = TQ \to Q$, that
is, the Poisson structure induced by the canonical symplectic
structure of $T^*Q$ (an explicit description of the linear Poisson
structure on $\bar{A}^* \simeq T^*Q/G$ may be found in
\cite{OrRa}).

Thus, $(\tilde{F}, F)$ is a linear Poisson morphism between $A^*=
T^*Q$ and $\bar{A}^* \simeq T^*Q/G$ and, in addition, $\tilde{F}$
is a fiberwise bijective vector bundle morphism.

Now, suppose that ${\mathcal G} = g$ is a $G$-invariant Riemannian
metric on $Q$ and that $V: Q \to \R$ is a $G$-invariant function
on $Q$. Then, we may consider the corresponding mechanical
Lagrangian function $L: A= TQ \to \R$. Moreover, it is clear that
$g$ and $V$ induce a bundle metric $\bar{\mathcal G}$ on $\bar{A}
= TQ/G$ and a real function $\bar{V}: \bar{Q} \to \R$ and,
therefore, a mechanical Lagrangian function $\bar{L}: \bar{A}=
TQ/G \to \R$.

On the other hand, we have that for each $q \in Q$ the map
$\tilde{F}_{q}: A^*_q = T^*_qQ \to \bar{A}^*_{F(q)} \simeq
(T^*Q/G)_{F(q)}$ is a linear isometry. Consequently, using
Corollary \ref{relHam-Jaun}, we deduce the following result
\begin{corollary}\label{Ham-JaAtun}
There exists a one-to-one correspondence between the $1$-cocycles
of the Atiyah algebroid $\tau_{\bar{A}}=\tau_{TQ/G}: \bar{A} =
TQ/G \to \bar{Q} = Q/G$ which are solutions of the Hamilton-Jacobi
equation for the mechanical Lagrangian function $\bar{L}: \bar{A}
= TQ/G \to \R$ and the $G$-invariant closed $1$-forms $\alpha$ on
$Q$ such that the function $\displaystyle \frac{1}{2}\|\#_{g}\circ
\alpha\|_{g}^2 + V$ is constant.
\end{corollary}

{\sl An explicit example}: {\sl \bf The Elroy's Beanie}. This
system is probably the most simple example of a dynamical system
with a non-Abelian Lie group of symmetries. It consists in two
planar rigid bodies attached at their centers of mass, moving
freely in the plane (see \cite{MaMoRa}). So, the configuration
space is $Q=SE(2)\times S^1$ with coordinates $q=(x, y, \theta,
\psi)$, where the three first coordinates  describe the position
and orientation of the center of mass of the first body and the
last one the relative orientation between both bodies. The
Lagrangian $L: TQ\to \R$ is
\[
L= \frac{1}{2}m
(\dot{x}^2+\dot{y}^2)+\frac{1}{2}
I_1\dot{\theta}^2+\frac{1}{2}I_2
(\dot{\theta}+\dot{\psi})^2-V(\psi)
\]
where $m$ denotes the mass of the system and
$I_1$ and $I_2$ are the inertias of the first and
the second body, respectively; additionally, we
also consider a potential function of the form
$V(\psi)$.
The kinetic energy is associated with the Riemannian metric ${\mathcal G}$ on $Q$ given by
\[
{\mathcal G}=m(dx^2+dy^2)+(I_1+I_2)d\theta^2+I_2d\theta\otimes
d\psi +I_2d\psi\otimes d\theta+I_2 d\psi^2.
\]

The system is $SE(2)$-invariant for the action  \[
\Phi_g(q)=\left( z_1+x\cos\alpha-y\sin \alpha,
z_2+x\sin\alpha+y\cos \alpha, \alpha+\theta, \psi \right)
\]
where $g=(z_1, z_2, \alpha)$.

Let $\{\xi_1, \xi_2, \xi_3\}$ be the standard basis of
$\frak{se}(2)$,
\[
{}[\xi_1,\xi_2]=0,\qquad [\xi_1,\xi_3]=-\xi_2\qquand [\xi_2,
\xi_3]=\xi_1\; .
\]
The quotient space $\bar{Q}=Q/SE(2)=(SE(2)\times S^1)/SE(2)\simeq
S^1$ is naturally parameterized by the coordinate $\psi$. The
Atiyah algebroid $TQ/SE(2)\to \bar{Q}$ is identified with the
vector bundle: $ \tau_{\bar{A}}: \bar{A}=TS^1\times {\mathfrak
se}(2)\to S^1. $ The canonical basis of sections of
$\tau_{\bar{A}}$ is: ${\displaystyle \left\{
\frac{\partial}{\partial \psi}, \xi_1, \xi_2, \xi_3\right\}}. $
Since the metric ${\mathcal G}$ is also $SE(2)$-invariant we
obtain a bundle metric $\bar{\mathcal G}$ and a $\bar{\mathcal
G}$-orthonormal basis of sections:
\[
\left\{
\displaystyle X_1=\sqrt{\frac{I_1+I_2}{I_1I_2}}\left(
\frac{\partial}{\partial \psi}-\frac{I_2}{I_1+I_2} \xi_3\right),
X_2=\frac{1}{\sqrt{m}}\xi_1, X_3=\frac{1}{\sqrt{m}}\xi_2, X_4=\frac{1}{\sqrt{I_1+I_2}}\xi_3\right\}
\]
In the coordinates $(\psi, v^1, v^2, v^3, v^4)$ induced by the orthonormal basis of sections, the reduced Lagrangian is $$\bar{L}=\frac{1}{2}\left( (v^1)^2+(v^2)^2 +(v^3)^2+(v^4)^2\right)-V(\psi)\; .$$

Additionally, we deduce that
\[
\begin{array}{ll}
\displaystyle \lcf X_1, X_2 \rcf_{\bar{A}}= -\sqrt{\frac{I_2}{I_1(I_1+I_2)}}X_3, & \displaystyle\lcf X_1, X_3 \rcf_{\bar{A}}=\sqrt{\frac{I_2}{I_1(I_1+I_2)}}X_2,\\[10pt]
\displaystyle \lcf X_1, X_4 \rcf_{\bar{A}}=0, &\displaystyle \lcf X_2, X_3 \rcf_{\bar{A}}=0,\\[10pt]
\displaystyle \lcf X_2, X_4 \rcf_{\bar{A}}=-\frac{1}{\sqrt{I_1+I_2}}X_3,& \displaystyle\lcf X_3, X_4 \rcf_{\bar{A}}=\frac{1}{\sqrt{I_1+I_2}}X_2.
\end{array}
\]
Therefore, the non-vanishing structure functions are
\[
\displaystyle C_{12}^3=-\sqrt{\frac{I_2}{I_1(I_1+I_2)}}, \quad
C_{13}^2=\sqrt{\frac{I_2}{I_1(I_1+I_2)}},\quad \displaystyle
C_{24}^3=-\frac{1}{\sqrt{I_1+I_2}}, \quad
C_{34}^2=\frac{1}{\sqrt{I_1+I_2}}.
\]
Moreover,
\[
\rho_{\bar{A}}(X_1)=\sqrt{\frac{I_1+I_2}{I_1I_2}}
\frac{\partial}{\partial \psi}, \quad \rho_{\bar{A}}(X_2)=0, \quad\rho_{\bar{A}}(X_3)=0,\quad\rho_{\bar{A}}(X_4)=0.
\]
The local expression of the Euler-Lagrange equations for the reduced Lagrangian system $\bar{L}:\bar{A}\to \R$ is:
\begin{eqnarray*}
\dot{\psi}&=&\sqrt{\frac{I_1+I_2}{I_1I_2}}v^1,\\
\dot{v}^1&=&-\sqrt{\frac{I_1+I_2}{I_1I_2}}
\frac{\partial V}{\partial \psi},\\
\dot{v}^2&=&-\sqrt{\frac{I_2}{I_1(I_1+I_2)}}v^1v^3+\frac{1}{\sqrt{I_1+I_2}}v^3v^4,\\
\dot{v}^3&=&\sqrt{\frac{I_2}{I_1(I_1+I_2)}}v^1v^2-\frac{1}{\sqrt{I_1+I_2}}v^2v^4,\\
\dot{v}^4&=&0.
\end{eqnarray*}
{}From the two first equations we obtain the equation:
\[
\ddot{\psi}=-\frac{I_1+I_2}{I_1I_2}
\frac{\partial V}{\partial \psi}.
\]
 A section  $\alpha: S^1\to \bar{A}^*$, $\alpha(\psi)=
 (\psi, \alpha_1(\psi), \alpha_2(\psi), \alpha_3(\psi), \alpha_4(\psi))$, is
 a 1-cocycle, i.e. $d^{\bar{A}}\alpha = 0$,
if and only  if
%\begin{eqnarray*}
%-\sqrt{\frac{I_2}{I_1(I_1+I_2)}}\alpha_3&=&\sqrt{\frac{I_1+I_2}{I_1I_2}}\frac{\partial \alpha_2}{\partial \psi}\\
%\sqrt{\frac{I_2}{I_1(I_1+I_2)}}\alpha_2&=&\sqrt{\frac{I_1+I_2}{I_1I_2}}\frac{\partial \alpha_3}{\partial \psi}\\
%0&=&\sqrt{\frac{I_1+I_2}{I_1I_2}}\frac{\partial \alpha_4}{\partial \psi}\\
%-\frac{1}{\sqrt{I_1+I_2}}\alpha_3&=&0\\
%\frac{1}{\sqrt{I_1+I_2}}\alpha_2&=&0
%\end{eqnarray*}
 $\alpha_2(\psi)=0$, $\alpha_3(\psi)=0$ and $\displaystyle\frac{\partial \alpha_4}{\partial \psi}=0$.
Therefore, the  Hamilton-Jacobi  equation $d^{\bar{A}}(E_{\bar{L}}
\circ \#_{\bar{\mathcal G}} \circ \alpha)=0$ is
\[
\frac{\partial V}{\partial \psi}+\frac{\partial \alpha_1}{\partial
\psi}\alpha_1=0.
\]
Thus,  integrating we obtain
\[
2V(\psi)+(\alpha_1(\psi))^2=k_1
\]
with $k_1$ constant. Therefore,
\[
\alpha_1(\psi)=\sqrt{k_1-2V(\psi)}
\]
and all the solutions of the Hamilton-Jacobi equation are of the form
\[
\alpha(\psi)=(\psi; \sqrt{k_1-2V(\psi)}, 0, 0, k_2).
\]
with $k_2$ constant.

\subsection{Mechanical systems subjected to linear nonholonomic
constraints on a Lie algebroid}

Let $\tau_{A}: A \to Q$ be a Lie algebroid over a manifold $Q$ and
denote by $(\lcf\cdot, \cdot \rcf_{A}, \rho_{A})$ the Lie
algebroid structure on $A$.

A \emph{mechanical system subjected to linear nonholonomic
constraints on $A$} is a pair $(L, D)$, where:
\begin{enumerate}
\item
$L: A \to \R$ is a \emph{Lagrangian function of mechanical type},
that is,
\[
L(a) = \displaystyle \frac{1}{2} {\mathcal G}(a, a) -
V(\tau_{A}(a)), \; \; \; \mbox{ for } a\in A,
\]
and
\item
$D$ is the total space of a vector subbundle $\tau_{D}: D \to Q$
of $A$. The vector subbundle $D$ is said to be the \emph{constraint subbundle}.

\end{enumerate}

This kind of systems were considered in \cite{CoLeMaMa,CoMa,GLMM}.

We will denote by $i_{D}: D \to A$ the canonical inclusion. We
also consider the orthogonal decomposition $A = D \oplus
D^{\perp}$ and the associated orthogonal projectors $P: A \to D$
and $Q: A \to D^{\perp}$. Then, the solutions of the dynamical
equations for the nonholonomic (constrained) system $(L, D)$ are
just the integral curves of the vector field $\xi_{(L, D)}$ on $D$
defined by
\[
\xi_{(L, D)} = TP \circ \xi_{L} \circ i_{D},
\]
where $\xi_{L}$ is the solution of the free dynamics (see Section
\ref{unconmesys}) and $TP: TA \to TD$ is the tangent map to the
projector $P$.

In fact, suppose that $(q^i)$ are local coordinates on an open
subset $U$ of $Q$ and that $\{X_{B}\} = \{X_{\gamma}, X_{b}\}$ is
a basis of sections of the vector bundle $\tau_{A}^{-1}(U) \to U$
such that $\{X_{\gamma}\}$ (respectively, $\{X_{b}\}$) is an
orthonormal basis of sections of the vector subbundle
$\tau_{D}^{-1}(U) \to U$ (respectively, $\tau_{D^{\perp}}^{-1}(U)
\to U$). We will denote by $(q^i, v^B) = (q^i, v^{\gamma}, v^b)$
the corresponding local coordinates on $A$. Then, the local
equations defining the vector subbundle $D$ are
\[
v^b = 0, \; \; \; \mbox{ for all } b.
\]
Moreover, if $\rho^i_{B}$ and $C^{E}_{BC}$ are the local structure
functions of $A$, we have that the local expression of the vector
field $\xi_{(L, D)}$ is
\begin{equation}\label{xiLD}
 \xi_{(L, D)} = \displaystyle
\rho^i_{\gamma}v^{\gamma} \frac{\partial}{\partial q^i} -
(C^{\nu}_{\delta\gamma} v^{\gamma}v^{\nu} + \rho^{i}_{\delta}
\frac{\partial V}{\partial q^i})\frac{\partial}{\partial
v^{\delta}}.
\end{equation}
Thus, the dynamical equations for the
constrained system $(L, D)$ are
\begin{equation}\label{E-Lnonh}
\dot{q}^i = \rho^{i}_{\gamma}v^{\gamma}, \; \; \; \dot{v}^{\delta}
= \displaystyle -C^{\nu}_{\delta\gamma} v^{\gamma}v^{\nu} -
\rho^{i}_{\delta} \frac{\partial V}{\partial q^i}, \; \; \; v^b =
0.
\end{equation}
On the other hand, the \emph{constrained connection} $\cnabla:
\Gamma(\tau_{A}) \times \Gamma(\tau_{A}) \to \Gamma(\tau_{A})$
associated with the system $(L, D)$ is given by
\[
\cnabla_{X}Y = P(\nabla^{\mathcal G}_{X}Y) + \nabla_{X}^{\mathcal
G}QY, \; \; \; \mbox{ for } X, Y \in \Gamma(\tau_{A}).
\]
Therefore, if $\check{\Gamma}_{BC}^{E}$ are the coefficients of
$\cnabla$, we have that
\[
\check{\Gamma}_{\gamma\nu}^{\delta} = \Gamma_{\gamma\nu}^{\delta}
= \displaystyle \frac{1}{2} (C_{\delta\gamma}^{\nu} +
C_{\delta\nu}^{\gamma} + C^{\delta}_{\gamma\nu}), \; \; \;
\check{\Gamma}_{\gamma\nu}^{a} = 0.
\]
Consequently, Eqs. (\ref{E-Lnonh}) are just the
\emph{Lagrange-D'Alembert equations} for the system $(L, D)$
considered in \cite{CoLeMaMa} (see also \cite{CoMa,GLMM}).

Next, we will introduce a linear almost Poisson structure
$\{\cdot, \cdot\}_{D^*}$ on $D^*$.

Denote by $\{\cdot, \cdot\}_{A^*}$ the linear Poisson bracket on
$A^*$ induced by the Lie algebroid structure on $A$. Then,
\begin{equation}\label{laPbDstar}
\{\varphi, \psi\}_{D^*} = \{\varphi \circ i_{D}^*,\psi \circ
i_{D}^*\}_{A^*} \circ P^*,
\end{equation}
for $\varphi, \psi \in C^{\infty}(D^*)$, where $i_{D}^*: A^* \to
D^*$ and $P^*: D^* \to A^*$ are the dual maps of the monomorphism
$i_{D}: D \to A$ and the projector $P: D \to A$, respectively.

It is easy to prove that $\{\cdot, \cdot\}_{D^*}$ is a linear
almost Poisson bracket on $D^*$. Moreover, if $(q^i, p_{B}) =
(q^i, p_{\gamma}, p_{b})$ are the dual coordinates of $(q^i, v^B)
= (q^i, v^{\gamma}, v^{b})$ on $A^*$ then it is clear that $(q^i,
p_{\gamma})$ are local coordinates on $D^*$ and, in addition, the
local expressions of $i_{D}^*$ and $P^*$ are
\begin{equation}\label{Localexp}
i_{D}^*(q^i, p_{\gamma}, p_{b}) = (q^i, p_{\gamma}), \; \; \;
P^*(q^i, p_{\gamma}) = (q^i, p_{\gamma}, 0).
\end{equation}
Thus, from (\ref{LambdaDstar}), (\ref{laPbDstar}) and
(\ref{Localexp}), we have that
\begin{equation}\label{laPbDstar1}
\{\varphi, \psi\}_{D^*} = \displaystyle
\rho^i_{\gamma}(\frac{\partial \varphi}{\partial q^i}
\frac{\partial \psi}{\partial p_{\gamma}} - \frac{\partial
\varphi}{\partial p_{\gamma}}\frac{\partial \psi}{\partial q^i}) -
C_{\beta\delta}^{\gamma}p_{\gamma} \frac{\partial
\varphi}{\partial p_{\beta}} \frac{\partial \psi}{\partial
p_{\delta}},
\end{equation}
for $\varphi, \psi \in C^{\infty}(D^*)$.

On the other hand, one may introduce a linear Poisson bracket
$\{\cdot, \cdot\}_{A}$ on $A$ in such a way that the vector bundle
map $\flat_{\mathcal G}: A \to A^*$ is a Poisson isomorphism, when
on $A^*$ we consider the linear Poisson structure $\{\cdot,
\cdot\}_{A^*}$. Since the local expression of $\flat_{\mathcal G}$
is
\[
\flat_{\mathcal G}(q^i, v^B) = (q^i, v^B)
\]
we deduce that the local expression of the linear Poisson bracket
$\{\cdot, \cdot\}_{A}$ is
\[
\{\bar{\varphi}, \bar{\psi}\}_{A} = \displaystyle
\rho^i_{B}(\frac{\partial \bar{\varphi}}{\partial q^i}
\frac{\partial \bar{\psi}}{\partial v^B} - \frac{\partial
\bar{\varphi}}{\partial v^{B}}\frac{\partial \bar{\psi}}{\partial
q^i}) - C_{BC}^{E}v^E \frac{\partial \bar{\varphi}}{\partial v^B}
\frac{\partial \bar{\psi}}{\partial v^C},
\]
for $\bar{\varphi}, \bar{\psi} \in C^{\infty}(A)$.

Using the bracket $\{\cdot, \cdot\}_{A}$, one may define a linear
almost Poisson bracket $\{\cdot, \cdot\}_{nh}$ on $D$ as follows.
If $\tilde{\varphi}$ and $\tilde{\psi}$ are real
$C^{\infty}$-functions on $D$ then
\[
\{\tilde{\varphi}, \tilde{\psi}\}_{nh} = \{\tilde{\varphi}\circ P,
\tilde{\psi}\circ P\}_{A} \circ i_{D}.
\]
We have that
\begin{equation}\label{nonholbrac}
\{\tilde{\varphi}, \tilde{\psi}\}_{nh} = \displaystyle
\rho^i_{\gamma}(\frac{\partial \tilde{\varphi}}{\partial q^i}
\frac{\partial \tilde{\psi}}{\partial v^\gamma} - \frac{\partial
\tilde{\varphi}}{\partial v^{\gamma}}\frac{\partial
\tilde{\psi}}{\partial q^i}) - C_{\beta \delta}^{\gamma}v^\gamma
\frac{\partial \tilde{\varphi}}{\partial v^\beta} \frac{\partial
\tilde{\psi}}{\partial v^\delta}.
\end{equation}
Thus, a direct computation proves that $\{\cdot, \cdot\}_{nh}$ is
just the \emph{nonholonomic bracket} introduced in
\cite{CoLeMaMa}. Note that, using (\ref{xiLD}) and
(\ref{nonholbrac}), we obtain that $\xi_{(L, D)}$ is the
Hamiltonian vector field of the function $(E_{L})_{|D}$ with
respect to the nonholonomic bracket $\{\cdot, \cdot\}_{nh}$, i.e.,
\[
\dot{\tilde{\varphi}} = \xi_{(L, D)}(\tilde{\varphi}) =
\{\tilde{\varphi}, (E_{L})_{|D}\}_{nh},
\]
for $\tilde{\varphi} \in C^{\infty}(D)$ (see also
\cite{CoLeMaMa}).

Moreover, if ${\mathcal G}_{D}$ is the restriction of the bundle
metric ${\mathcal G}$ to $D$ and $\flat_{{\mathcal G}_{D}}: D \to
D^*$ is the corresponding vector bundle isomorphism then, from
(\ref{laPbDstar1}) and (\ref{nonholbrac}), we deduce that
\[
\{\varphi \circ \flat_{{\mathcal G}_{D}}, \psi \circ
\flat_{{\mathcal G}_{D}}\}_{nh} = \{\varphi, \psi\}_{D^*} \circ
\flat_{{\mathcal G}_{D}}, \; \; \; \mbox{ for } \varphi, \psi \in
C^{\infty}(D^*).
\]
For this reason, $\{\cdot, \cdot\}_{D^*}$ will also be called
the \emph{nonholonomic bracket associated with the constrained
system} $(L, D)$.

We will denote by $(\lcf\cdot, \cdot\rcf_{D}, \rho_{D})$
(respectively, $d^D$) the corresponding skew-symmetric algebroid
structure (respectively, almost differential) on the vector bundle
$\tau_{D}: D\to Q$ and by $\#_{{\mathcal G}_{D}}: D^* \to D$ the
inverse morphism of $\flat_{{\mathcal G}_{D}}: D \to D^*$.

Then, from (\ref{alafromalp}) and (\ref{laPbDstar}), it follows
that
\begin{equation}\label{alLieD}
\lcf X, Y\rcf_{D} = P\lcf i_{D}\circ X, i_{D}\circ Y\rcf_{A}, \;
\; \; \rho_{D}(X) = \rho_{A}(i_{D}\circ X),
\end{equation}
for $X, Y \in \Gamma(\tau_{D})$. Therefore, using (\ref{deD}), we
have that
\begin{equation}\label{difD}
d^D\alpha = \Lambda^k i_{D}^*(d^A(P^* \circ \alpha)), \; \; \;
\mbox{ for } \alpha \in \Gamma(\Lambda^k\tau_{D^*}).
\end{equation}
On the other hand, if $\alpha: Q \to D^*$ is a section of the
vector bundle $\tau_{D^*}: D^* \to Q$ one may consider the vector
field $\xi_{(L, D),\alpha}$ on $Q$ given by
\begin{equation}\label{xiLDalpha}
\xi_{(L, D)\alpha}(q) = (T_{\#_{{\mathcal
G}_{D}}(\alpha(q))}\tau_{D})(\xi_{(L, D)}(\#_{{\mathcal
G}_{D}}(\alpha(q)))), \; \; \; \mbox{ for } q\in Q.
\end{equation}

\begin{corollary}\label{nonhoHam-Jaeq}
Let $\alpha: Q \to D^*$ be a $1$-cocycle of the skew-symmetric
algebroid $(D, \lcf\cdot, \cdot\rcf_{D}, \rho_{D})$, that is,
$d^D\alpha = 0$. Then, the following conditions are equivalent:
\begin{enumerate}
\item
If $c: I \to Q$ is an integral curve of the vector field $\xi_{(L,
D)\alpha}$ on $Q$ we have that $\#_{{\mathcal G}_{D}}\circ \alpha
\circ c: I \to D$ is a solution of the Lagrange-D'Alembert
equations for the constrained system $(L, D)$.
\item
$\alpha$ satisfies \emph{the nonholonomic Hamilton-Jacobi
equation}
\[
d^D((E_{L})_{|D} \circ \#_{{\mathcal G}_{D}} \circ \alpha) = 0.
\]
\noindent If, additionally, $H^0(d^D) \simeq \R$ (or the
skew-symmetric algebroid $(D, \lcf \cdot , \cdot \rcf_{D},
\rho_{D})$ is completely nonholonomic and $Q$ is connected) then
conditions (i) and (ii) are equivalent to
\item
$(E_{L})_{|D} \circ \#_{{\mathcal G}_{D}} \circ \alpha = \mbox{
constant }.$
\end{enumerate}
\end{corollary}
\begin{proof}
Denote by $h_{(L, D)}: D^* \to \R$ the Hamiltonian function on
$D^*$ given by $h_{(L, D)} = (E_{L})_{|D} \circ \#_{{\mathcal
G}_{D}}$, by $\Lambda_{D^*}$ the linear almost Poisson $2$-vector
on $D^*$ and by ${\mathcal H}_{h_{(L, D)}}^{\Lambda_{D^*}}$ the
Hamiltonian vector field of $h_{(L, D)}$ with respect to
$\Lambda_{D^*}$. Then, the vector fields $\xi_{(L, D)}$ and
${\mathcal H}_{h_{(L, D)}}^{\Lambda_{D^*}}$ on $D$ and $D^*$,
respectively, are $\flat_{{\mathcal G}_{D}}$-related. Thus, from
(\ref{xiLDalpha}) and since $\tau_{D^*} \circ \flat_{{\mathcal
G}_{D}} = \tau_{D}$, it follows that
\[
{\mathcal H}_{h_{(L, D)}\alpha}^{\Lambda_{D^*}}(q) =
(T_{\alpha(q)}\tau_{D^*})({\mathcal H}_{h_{(L,
D)}}^{\Lambda_{D^*}}(\alpha(q))) = \xi_{(L, D)\alpha}(q),
\]
that is, the vector fields ${\mathcal H}_{h_{(L,
D)}\alpha}^{\Lambda_{D^*}}$ and $\xi_{(L, D)\alpha}$ are equal.

Moreover, if $\sigma: I \to D$ is a curve on $D$, we have that
$\sigma$ is a solution of the Lagrange-D'Alembert equations for
the constrained system $(L, D)$ if and only if $\flat_{{\mathcal
G}_{D}} \circ \sigma: I \to D^*$ is a solution of the Hamilton
equations for $h_{(L, D)}$.

Therefore, using Theorem \ref{maintheorem}, we deduce that
conditions (i) and (ii) are equivalent.

In addition, if $H^0(d^D) \simeq \R$ (or if $(D, \lcf \cdot ,
\cdot \rcf_{D}, \rho_{D})$ is completely nonholonomic and $Q$ is
connected) then, from Corollary \ref{strong-HJ}, it follows that
conditions (i), (ii) and (iii) are equivalent.
\end{proof}

\begin{remark}\label{relIgLeMa}{\rm
Let $D^0$ be the annihilator of $D$ and ${\mathcal I}(D^0)$ be the
algebraic ideal generated by $D^0$. Thus, a section $\nu$ of the
vector bundle $\Lambda^kA^* \to Q$ belongs to ${\mathcal I}(D^0)$
if
\[
\nu(q)(v_{1}, \dots , v_{k}) = 0, \; \; \mbox{ for all } q\in Q
\mbox{ and } v_{1}, \dots , v_{k} \in D_{q}.
\]
Now, let ${\mathcal Z}(\tau_{D^*})$ be the set defined by
\[
{\mathcal Z}(\tau_{D^*}) = \{ \alpha \in \Gamma(\tau_{D^*}) /
d^D\alpha = 0\}
\]
and $\tilde{\mathcal Z}(\tau_{(D^{\perp})^0})$ be the set given by
\[
\tilde{\mathcal Z}(\tau_{(D^{\perp})^0}) = \{ \tilde{\alpha} \in
\Gamma(\tau_{(D^{\perp})^0}) / d^A\tilde{\alpha} \in {\mathcal
I}(D^0)\}
\]
where $(D^{\perp})^0$ is the annihilator of the orthogonal
complement $D^{\perp}$ of $D$ and $\tau_{(D^{\perp})^0}:
(D^{\perp})^0 \to Q$ is the corresponding vector bundle
projection. Then, using (\ref{difD}), we deduce that the map
\[
{\mathcal Z}(\tau_{D^*}) \to \tilde{\mathcal
Z}(\tau_{(D^{\perp})^0}), \; \; \; \alpha \to P^* \circ \alpha
\]
defines a bijection from ${\mathcal Z}(\tau_{D^*})$ on
$\tilde{\mathcal Z}(\tau_{(D^{\perp})^0})$. In fact, the inverse
map is given by
\[
\tilde{\mathcal Z}(\tau_{(D^{\perp})^0}) \to {\mathcal
Z}(\tau_{D^*}), \; \; \; \tilde{\alpha} \to i_{D}^* \circ
\tilde{\alpha}.
\]
On the other hand, if $f$ is a real $C^{\infty}$-function on $Q$
then
\[
d^Df = 0 \Longleftrightarrow (d^Af)(Q) \subseteq D^0.
\]
}
\end{remark}

Let $(L, D)$ (respectively, $(\bar{L}, \bar{D})$) be a
nonholonomic system on a Lie algebroid $\tau_{A}: A \to Q$
(respectively, $\tau_{\bar{A}}: \bar{A} \to \bar{Q}$) and
$(\tilde{F}, F)$ be a linear almost Poisson morphism between the
almost Poisson manifolds $(D^*, \{\cdot, \cdot\}_{D^*})$ and
$(\bar{D}^*, \{\cdot, \cdot\}_{\bar{D}^*})$ such that:
\begin{enumerate}
\item
$F: Q \to \bar{Q}$ is a surjective map.
\item
For each $q\in Q$, the linear map $\tilde{F}_{q} =
\tilde{F}_{|D^*_q}: D^*_q \to \bar{D}^*_{F(q)}$ satisfies the
following conditions
\[
{\mathcal G}_{\bar{D}^*}(\tilde{F}_{q}(\beta),
\tilde{F}_{q}(\beta')) = {\mathcal G}_{D^*}(\beta, \beta'), \; \;
\; \mbox{ for } \beta, \beta' \in D^*_q,
\]
\vspace{-.5cm}
\[
F(q) = F(q') \Longrightarrow \tilde{F}_q(D^*_q) =
\tilde{F}_{q'}(D^*_{q'}),
\]
where ${\mathcal G}_{D^*}$ (respectively, ${\mathcal
G}_{\bar{D}^*}$) is the bundle metric on $D^*$ (respectively,
$\bar{D}^*$).
\item
If $V: Q \to \R$ (respectively, $\bar{V}: \bar{Q} \to \R$) is the
potential energy for the nonholonomic system on $A$ (respectively,
$\bar{A}$) we have that $\bar{V} \circ F = V$.
\end{enumerate}
Then, we deduce that $(\tilde{F}, F)$ is a Hamiltonian morphism
between the Hamiltonian systems $(D, \{\cdot, \cdot\}_{D^*},
\linebreak h_{(L, D)})$ and $(\bar{D}, \{\cdot,
\cdot\}_{\bar{D}^*}, h_{(\bar{L}, \bar{D})})$, where $h_{(L, D)}$
(respectively, $h_{(\bar{L}, \bar{D})}$) is the constrained
Hamiltonian function on $D^*$ (respectively, $\bar{D}^*$)
associated with the nonholonomic system $(L, D)$ (respectively,
$(\bar{L}, \bar{D})$).

Moreover, using Theorem \ref{morHa-Ja}, we conclude that
\begin{corollary}\label{relHam-Jaun-1}
\begin{enumerate}
\item
If $\alpha: Q \to D^*$ is a $1$-cocycle for the skew-symmetric
algebroid $D$ ($d^D\alpha = 0$), it satisfies the Hamilton-Jacobi
equation
\begin{equation}\label{EqHa-Janon1}
d^D((E_{L})_{|D} \circ \#_{{\mathcal G}_{D}} \circ \alpha) = 0
\end{equation}
(respectively, the strongest condition $(E_{L})_{|D} \circ
\#_{{\mathcal G}_{D}} \circ \alpha = \mbox{ constant }$) and it is
$(\tilde{F}, F)$-related with $\bar{\alpha} \in
\Gamma(\tau_{\bar{D}^*})$ then $d^{\bar{D}}\bar{\alpha} = 0$ and
$\bar{\alpha}$ is a solution of the Hamilton-Jacobi equation
\begin{equation}\label{EqHa-Janon2}
d^{\bar{D}}((E_{\bar{L}})_{|\bar{D}} \circ \#_{{\mathcal
G}_{\bar{D}}} \circ \bar{\alpha}) = 0
\end{equation}
(respectively, $\bar{\alpha}$ satisfies the strongest condition
$(E_{\bar{L}})_{|\bar{D}} \circ \#_{{\mathcal G}_{\bar{D}}} \circ
\bar{\alpha} = \mbox{ constant }$).
\item
If $\bar{\alpha}: \bar{Q} \to \tilde{F}(D^*) \subseteq \bar{D}^*$
is a $1$-cocycle for the skew-symmetric algebroid $\bar{D}$
($d^{\bar{D}}\bar{\alpha} = 0$) and it satisfies the
Hamilton-Jacobi equation (\ref{EqHa-Janon2}) (respectively, the
strongest condition $(E_{\bar{L}})_{|\bar{D}} \circ \#_{{\mathcal
G}_{\bar{D}}} \circ \bar{\alpha} = \mbox{ constant }$) then
$d^{D}\alpha = 0$ and $\alpha$ is a solution of the Hamilton
Jacobi equation (\ref{EqHa-Janon1}) (respectively, $\alpha$
satisfies the strongest condition $(E_{L})_{|D} \circ
\#_{{\mathcal G}_{D}} \circ \alpha = \mbox{ constant }$). Here,
$\alpha: Q \to D^*$ is the section of $\tau_{D^*}: D^* \to Q$
characterized by the condition $\tilde{F} \circ \alpha =
\bar{\alpha} \circ F$.
\end{enumerate}
\end{corollary}

\subsubsection{The particular case $A = TQ$} Let $L: TQ \to \R$ be
a Lagrangian function of mechanical type on the standard Lie
algebroid $\tau_{TQ}: TQ \to Q$, that is,
\[
L(v) = \displaystyle \frac{1}{2} g(v, v) - V(\tau_{Q}(v)), \; \;
\; \mbox{ for } v \in TQ,
\]
where $g$ is a Riemannian metric on $Q$ and $V: Q \to \R$ is a
real $C^{\infty}$-function on $Q$. Suppose also that $D$ is a
distribution on $Q$. Then, the pair $(L, D)$ is a mechanical
system subjected to linear nonholonomic constraints on the
standard Lie algebroid $\tau_{TQ}: TQ \to Q$.

Note that, in this case, the linear Poisson structure on $A^* =
T^*Q$ is induced by the canonical symplectic structure on $T^*Q$.
Moreover, the corresponding nonholonomic bracket $\{\cdot,
\cdot\}_{D^*}$ on $D^*$ was considered by several authors or,
alternatively, other almost Poisson structures (on $D$ or on
$\flat_{g}(D) \subseteq A^* = T^*Q$) which are isomorphic to
$\{\cdot, \cdot\}_{D^*}$ also were obtained by several authors
(see \cite{CaLeMa,IbLeMaMa,KooMa,VaMa}).

Now, denote by $\#_{g}: T^*Q \to TQ$ (respectively, $\#_{g_{D}}:
D^* \to D$) the inverse morphism of the musical isomorphism
$\flat_{g}: TQ \to T^*Q$ (respectively, $\flat_{g_{D}}: D \to
D^*$) induced by the Riemannian metric $g$ (respectively, by the
restriction $g_{D}$ of $g$ to $D$), by $d$ the standard exterior
differential on $Q$ (that is, $d = d^{TQ}$ is the differential of
the Lie algebroid $\tau_{TQ}: TQ \to Q$), by $\xi_{(L, D)}\in
{\frak X}(D)$ the solution of the nonholonomic dynamics and by
$\xi_{(L, D)\alpha} \in {\frak X}(Q)$ its projection on $Q$,
$\alpha$ being a section of the vector bundle $\tau_{D^*}: D^* \to
Q$ (see (\ref{xiLDalpha})). Using this notation, Corollary
\ref{nonhoHam-Jaeq} and Remark \ref{relIgLeMa}, we deduce the
following result
\begin{corollary}\label{nonHa-Jastan}
Let $\alpha: Q \to D^*$ be a section of the vector bundle
$\tau_{D^*}: D^* \to Q$ such that $d(P^* \circ \alpha) \in
{\mathcal I}(D^0)$. Then, the following conditions are equivalent:
\begin{enumerate}
\item
If $c: I \to Q$ is an integral curve of the vector field $\xi_{(L,
D)\alpha}$ on $Q$ we have that $\#_{g_{D}} \circ \alpha \circ c: I
\to D$ is a solution of the Lagrange-D'Alembert equations for the
constrained system $(L, D)$.
\item
$d((E_{L})_{|D} \circ \#_{g_{D}} \circ \alpha)(Q) \subseteq D^0$.
\end{enumerate}
\end{corollary}
\begin{remark}{\rm
As we know, the Legendre transformation associated with the
Lagrangian function $L: TQ \to \R$ is the musical isomorphism
$\flat_{g}: TQ \to T^*Q$. Moreover, it is clear that $X(Q)
\subseteq D$, where $X$ is the vector field on $Q$ given by $X=
\#_{g_{D}} \circ \alpha$. Thus, Corollary \ref{nonHa-Jastan} is a
consequence of some results which were proved in \cite{IgLeMa}
(see Theorem 4.3 in \cite{IgLeMa}). On the other hand, if
$H^0(d^D) \simeq \R$ (or if $Q$ is connected and the distribution
$D$ is completely nonholonomic in the sense of Vershik and
Gershkovich \cite{VeGe}) then (i) and (ii) in Corollary
\ref{nonHa-Jastan} are equivalent to the condition
\[
(E_{L})_{|D} \circ \#_{g_{D}} \circ \alpha = \mbox{ constant }.
\]
A Hamiltonian version of this last result was proved by Ohsawa and
Bloch \cite{OhBl} (see Theorem 3.1 in \cite{OhBl}).
 }
\end{remark}

\begin{remark}\label{history}{\bf Previous approaches.}
{\rm
There exists some different attempts in the literature of extending
the classical Hamilton-Jacobi equation for the case of nonholonomic constraints
\cite{eden,Pa,Rum,Doo1,Doo2,Doo3,Doo4}).
These attempts  were non-effective or very restrictive (and even erroneous),
because, in many of them, they try to adapt the standard proof of
the Hamilton-Jacobi equations for systems without constraints,
using Hamilton's principle. See \cite{Sum} for a detailed discussion on the topic.

To fix ideas, consider a lagrangian system $L: TQ\longrightarrow
\R$ of mechanical type, that is, $L(v_q)=\frac{1}{2}{\mathcal
G}(v_q, v_q)-V(q)$, for $v_q\in T_qQ$,  and nonholonomic
constraints  determined by a distribution $D$ of $Q$, whose
annihilator is $ D^0=\hbox{span}\{\mu^b_i\, dq^i\}$.

 The  idea of many of these previous approaches consist in
 looking for a function $S: Q\longrightarrow \R$ called
 the characteristic function which permits characterize
 the solutions of the nonholonomic problem.  For it,
 define first the  generalized  momenta
   \[
 p_i=\frac{\partial S}{\partial q^i}+\lambda_{b}\mu^b_i\; ,
 \]
 which satisfy the constraint equations ${\mathcal G}^{ij}p_i\mu^a_j=0$.
 These last conditions univocally determine $\lambda_b$ as
 functions of $q$ and $\partial S/\partial q$ and therefore we find the momenta as
 functions
\begin{equation}\label{puo}
 p_i=p_i(q^i, \frac{\partial S}{\partial q^i})\; .
 \end{equation}
 By inserting these expressions for the generalized momenta in the Hamiltonian of the system,
 we obtain a version of the Hamilton-Jacobi equation (in its time-independent version):
 \begin{equation}\label{415}
   H(q^i, p_i)=\tilde{H}(q^i, \frac{\partial S}{\partial q^i})=\hbox{constant.}
 \end{equation}
  However, if we start with a curve $c: I\to Q$ satisfying the differential equations
 \begin{equation}\label{pua1}
 \dot{c}^i(t)=\frac{\partial H}{\partial p_i} (c^j(t),
 \frac{\partial S}{\partial q^j}(c(t))+\lambda_b\mu^b_j)
\end{equation}
in general, it is not true that the curve $\gamma(t)=(c^i(t), p_i(t))$
is a solution of the nonholonomic equations.
This is trivially checked since from Equation (\ref{415}) we deduce that:
\begin{equation}\label{416}
0=\frac{\partial H}{\partial q^i}+\frac{\partial H}{\partial p_j}
\left[\frac{\partial^2 S}{\partial q^i\partial q^j}+
\frac{\partial \lambda_b}{\partial q^i}\mu^b_j+\lambda_b \frac{\partial \mu^b_j}{\partial q^i}\right]
\end{equation}
but, on the other hand,
\begin{eqnarray}
\dot{p}_i&=&\frac{d}{dt}\left[\frac{\partial S}{\partial q^i}+\lambda_{b}\mu^b_i\right]\; ,\nonumber \\
&=& \dot{q}^j\left[\frac{\partial^2 S}{\partial q^i\partial q^j}+\frac{\partial \lambda_b}{\partial q^j}\mu^b_i+\lambda_b \frac{\partial \mu^b_i}{\partial q^j}\right]\label{417}.
\end{eqnarray}

Substituting Equation (\ref{416}) in Equation (\ref{417}), a curve
$\gamma(t)=(c^i(t), p_i(t))$ satisfying (\ref{pua1}) is solution
of the nonholonomic equations (that is, $\dot{p}_i=-\frac{\partial
H}{\partial q^i}+\Lambda_b\mu^b_i$) if it verifies the following
condition:
\begin{equation}\label{pua2}
\lambda_b\left(\frac{\partial \mu^b_j}{\partial
q^i}\dot{q}^j-\frac{\partial \mu^b_i}{\partial
q^j}\dot{q}^j\right)\delta q^i=0, \quad  \delta q\in D_q.
\end{equation}
It is well-known (see \cite{Rum, Sum}) that condition (\ref{pua2})
takes place when the solutions of the nonholonomic problem are
also of variational type.
  However,
nonholonomic dynamics is not, in general,  of variational kind (see
\cite{CoLeMaSo,LeMaDa,LeMu}). Indeed, a relevant difference with
the unconstrained mechanical systems is that a nonholonomic system
is not Hamiltonian in the standard sense since the dynamics is
obtained from  an almost Poisson bracket, that is, a bracket not
satisfying the Jacobi identity (see
\cite{CaLeMa,IbLeMaMa,KooMa,VaMa}).
}
\end{remark}

\noindent {\it An explicit example:} \emph{The two-wheeled
carriage} (see \cite{NeFu}). The system has configuration space
$Q=SE(2)\times \T^2$ , where $SE(2)$ represents the rigid motions
in the plane and $\T^2$ the angles of rotation of the left and
right wheels. We use standard coordinates $(x,y, \theta, \psi_1,
\psi_2)\in SE(2)\times \T^2$. Imposing the constraints of no
lateral sliding and no sliding on both wheels, one gets the
following nonholonomic constraints:
\begin{eqnarray*}
\dot{x}\sin \theta - \dot{y} \cos\theta &=& 0, \\
\dot{x}\cos\theta + \dot{y} \sin\theta+r\dot\theta+a\dot{\psi}_1 &=& 0, \\
\dot{x}\cos\theta + \dot{y} \sin\theta-r\dot\theta+a\dot{\psi}_2
&=& 0,
\end{eqnarray*}
where $a$ is the radius of the wheels and $r$ is the half the
length of the axle.

Assuming, for simplicity, that the center of mass of the carriage
is situated on the center of the axle the Lagrangian is given by:
\[
L=\displaystyle \frac{1}{2}m\dot{x}^2 + \frac{1}{2} m\dot{y}^2 +
\frac{1}{2}J\dot{\theta}^2 +
\frac{1}{2}C\dot{\psi}_1^2+\frac{1}{2}C\dot{\psi}_2^2,
\]
where $m$ is the mass of the system, $J$ the moment of inertia
when it rotates as a whole about the vertical axis passing through
the point $(x,y)$ and $C$ the axial moment of inertia.
 Note
that $L$ is the kinetic energy associated with the Riemannian
metric $g$ on $Q$ given by
\[
g = m(dx^2 + dy^2) + Jd\theta^2 + Cd\psi_1^2+Cd\psi_2^2\; .
\]
The constraints induce the distribution $D$ locally spanned by the
following $g$-orthonormal vector fields
\begin{eqnarray*}
 X_1&=&\frac{1}{\Lambda_1}\left(2r\frac{\partial}{\partial \psi_1}-a\frac{\partial}{\partial \theta}-ar \cos\theta\frac{\partial}{\partial x}
-ar \sin\theta\frac{\partial}{\partial y}\right),\\
 X_2&=&\frac{1}{\Lambda_2}\left(
 a^2(J-m_1r^2)\frac{\partial}{\partial \psi_1}+(a^2J+4Cr^2+a^2m_1r^2)\frac{\partial}{\partial \psi_2}+ar(2C+m_1a^2)\frac{\partial}{\partial \theta}\right.\\
 &&\left.-a(a^2J+2Cr^2)\cos\theta\frac{\partial}{\partial x}
-a(a^2J+2Cr^2)\sin\theta\frac{\partial}{\partial y} \right)\; .
\end{eqnarray*}
where
\begin{eqnarray*}
\Lambda_1&=& \sqrt{4Cr^2+a^2J+a m_1r^2}\\
\Lambda_2&=& \sqrt{(a^2J+2Cr^2)(2C+m_1a^2)(a^2J+4Cr^2+a^2r^2m_1)}
\end{eqnarray*}
 We will denote
by $(x, y, \theta, \psi_1, \psi_2, v^1, v^2)$ the local
coordinates on $D$ induced by the basis $\{X_{1}, X_{2}\}$.

In these coordinates, the restriction, $L_{|D}: D\longmapsto \R$,
of $L$ to $D$ is:
\[
L_{|D}=\frac{1}{2}((v^1)^2+(v^2)^2)\; .
\]

The distribution $D^\perp$ orthogonal to $D$ is generated by
\[
D^\perp=\{ X_3=\tan\theta \frac{\partial}{\partial
x}-\frac{\partial}{\partial y},
X_4=\frac{J\sec\theta}{rm_1}\frac{\partial}{\partial
x}+\frac{\partial}{\partial \theta}+
\frac{aJ}{Cr}\frac{\partial}{\partial \psi_1},
X_5=\frac{2C\sec\theta}{am_1}\frac{\partial}{\partial
x}+\frac{\partial}{\partial \psi_1}+\frac{\partial}{\partial
\psi_2}\}
\]

Moreover, since the standard Lie bracket $[X_{1}, X_{2}]$ of the
vector fields $X_{1}$ and $X_{2}$ is orthogonal to $D$, it follows
that (see (\ref{alLieD}))
\begin{equation}\label{Eqsestejem}
\lcf X_{1}, X_{2}\rcf_{D} = 0, \; \; \rho_{D}(X_{1}) = X_{1}, \;
\; \rho_{D}(X_{2}) = X_{2},
\end{equation}
where $(\lcf\cdot, \cdot\rcf_{D}, \rho_{D})$ is the skew-symmetric
algebroid structure on the vector bundle $\tau_{D}: D \to Q$.

The local expression of the vector field $\xi_{(L,D)}$ is:
\begin{eqnarray*}
\xi_{(L,D)}&=&\left(\frac{2rv^1}{\Lambda_1}+\frac{a^2(J-m_1r^2)v^2}{\Lambda_2}\right)\frac{\partial}{\partial
\psi_1}
+\frac{(a^2J+4Cr^2+a^2m_1r^2)v^2}{\Lambda_2}\frac{\partial}{\partial
\psi_2}\\
&&+\left(\frac{ar(2C+m_1a^2)v^2}{\Lambda_2}-\frac{av_1}{\Lambda_1}\right)\frac{\partial}{\partial \theta} - \left(\frac{ar
v^1\cos\theta}{\Lambda_1}+\frac{a(a^2J+2Cr^2)v^2\cos\theta
}{\Lambda_2}\right)\frac{\partial}{\partial x}\\&& -\left(\frac{ar
v^1\sin\theta}{\Lambda_1}+\frac{a(a^2J+2Cr^2)v^2\sin\theta
}{\Lambda_2}\right)\frac{\partial}{\partial y}
\end{eqnarray*}

Furthermore, if $\{X^{1}, X^{2}\}$ is the dual basis of $\{X_{1},
X_{2}\}$ and $\alpha: Q \to D^*$ is a section of the vector bundle
$\tau_{D^*}: D^* \to Q$
\[
\alpha = \alpha_{1}X^1 + \alpha_{2}X^2, \; \; \mbox{ with }
\alpha_{1}, \alpha_{2} \in C^{\infty}(Q)
\]
then
\[
\alpha \mbox{ is a 1-cocycle } \Longleftrightarrow
X_{1}(\alpha_{2}) - X_{2}(\alpha_{1}) = 0.
\]

In particular, taking
\[
\alpha = K_{1}X^1 + K_{2}X^2, \; \; \mbox{ with } K_{1}, K_{2} \in
\R
\]
trivially is satisfied the  $1$-cocycle condition.

In addition, since $E_{L} = L$, we deduce that
\[
(E_{L})_{|D} = \displaystyle \frac{1}{2} ((v^1)^2 + (v^2)^2)
\]
which implies that
\[
(E_{L})_{|D} \circ \#_{g_{D}} \circ \alpha =
K_1^2+K_2^2=\hbox{constant}.
\]
Thus, using Corollary \ref{nonhoHam-Jaeq}, we conclude that to
integrate the nonholonomic mechanical system $(L, D)$ is
equivalent to find the integral curves of the vector field on $Q =
S^1 \times S^1 \times \R^2$ given by
\[
\xi_{(L, D),\alpha} = K_1 X_1+K_2 X_2\; .
\]
which are easily obtained.

It is also interesting to observe that, in this particular
example,
\[
\mathrm{Lie}^{\infty}(D)=\{ad\psi^1-ad\psi^2+2rd\theta\}^0
\]
and, thus, $D$ is not completely nonholonomic. From Theorem 3.4 it
is necessary to restrict the initial nonholonomic system to the
orbits of $D$, that in this case are
\[
L_k=\{ (x, y,\theta, \psi_1, \psi_2)\in SE(2)\times \T^2\; |\;
a(\psi_1-\psi_2)+2r\theta=k, \hbox{  with  } k\in\R\}
\]
to obtain a completely nonholonomic skew-symmetric algebroid
structure on the vector bundle $\tau_{D_{L_{k}}}: D_{L_{k}} \to
L_{k}$. Note that on $L_{k}$ we can use, for instance, coordinates
$(x,y, \psi_1, \psi_2)$.

\subsubsection{The particular case $\bar{A} = TQ/G$}
Let $F: Q \to \bar{Q} = Q/G$ be a principal $G$-bundle and
$\tau_{\bar{A}} = \tau_{TQ/G}: \bar{A} = TQ/G \to \bar{Q} = Q/G$
be the Atiyah algebroid associated with the principal bundle (see
Section \ref{unconmesys}).

Suppose that $g$ is a $G$-invariant Riemannian metric on $Q$, that
$V: Q \to \R$ is a $G$-invariant real $C^{\infty}$-function and
that $D$ is a $G$-invariant distribution on $Q$. Then, we may
consider the corresponding nonholonomic mechanical system $(L, D)$
on the standard Lie algebroid $\tau_{A} = \tau_{TQ}: A= TQ \to Q$.

Denote by $\xi_{(L, D)} \in {\frak X}(D)$ the nonholonomic
dynamics for the system $(L, D)$ and by $\{\cdot, \cdot\}_{D^*}$
the nonholonomic bracket on $D^*$.

The Riemannian metric $g$ and the function $V: Q \to \R$ induce a
bundle metric $\bar{\mathcal G}$ on the Atiyah algebroid
$\tau_{\bar{A}} = \tau_{TQ/G}: \bar{A} = TQ/G \to \bar{Q} = Q/G$
and a real $C^{\infty}$-function $\bar{V}: \bar{Q} \to \R$ on
$\bar{Q}$ such that $\bar{V}\circ F = V$, where $F: Q \to \bar{Q}
= Q/G$ is the canonical projection. Moreover, the space of orbits
$\bar{D}$ of the action of $G$ on $D$ is a vector subbundle of the
Atiyah algebroid $\tau_{\bar{A}} = \tau_{TQ/G}: \bar{A} = TQ/G \to
\bar{Q} = Q/G$. Thus, we may consider the corresponding
nonholonomic mechanical system $(\bar{L}, \bar{D})$ on $\bar{A} =
TQ/G$.

Let $\bar{F}: A=TQ \to \bar{A} = TQ/G$ be the canonical
projection. Then, $(\bar{F}, F)$ is a fiberwise bijective morphism
of Lie algebroids and $\bar{F}(D) = \bar{D}$. Therefore, using
some results in \cite{CoLeMaMa} (see Theorem 4.6 in
\cite{CoLeMaMa}) we deduce that the vector field $\xi_{(L, D)}$ is
$\bar{F}_{D}$-projectable on the nonholonomic dynamics
$\xi_{(\bar{L}, \bar{D})} \in {\frak X}(\bar{D})$ of the system
$(\bar{L}, \bar{D})$. Here, $\bar{F}_{D}: D \to \bar{D} = D/G$ is
the canonical projection.

On the other hand, if $P: A = TQ \to D$ and $\bar{P}: \bar{A} =
TQ/G \to \bar{D} = D/G$ are the orthogonal projectors then it is
clear that
\[
\bar{F}_{D} \circ P = \bar{P} \circ \bar{F}
\]
which implies that
\begin{equation}\label{Importante}
\tilde{F} \circ P^* = \tilde{F}_{D} \circ \bar{P}^*,
\end{equation}
where $\tilde{F}: A^* = T^*Q \to \bar{A}^*\simeq T^*Q/G$ and
$\tilde{F}_{D}: D^* \to \bar{D}^*\simeq D^*/G$ are the canonical
projections. Moreover, if on $\bar{A}^*$ we consider the linear
Poisson structure induced by the Atiyah algebroid $\tau_{\bar{A}}
= \tau_{TQ/G}: \bar{A} = TQ/G \to \bar{Q} = Q/G$ then, as we know,
$\tilde{F}: A^* = T^*Q \to \bar{A}^*\simeq T^*Q/G$ is a Poisson
morphism. Thus, using this fact, (\ref{laPbDstar}) and
(\ref{Importante}), we deduce the following result
\begin{proposition}\label{nonhobraAtiyah}
The pair $(\tilde{F}_{D}, F)$ is a linear almost Poisson morphism,
when on $D^*$ and $\bar{D}^*$ we consider the almost Poisson
structures induced by the nonholonomic brackets $\{\cdot,
\cdot\}_{D^*}$ and $\{\cdot, \cdot\}_{\bar{D}^*}$, respectively.
\end{proposition}
Note that Proposition \ref{nonhobraAtiyah} characterizes the
nonholonomic bracket $\{\cdot, \cdot\}_{\bar{D}^*}$.

We also note that the linear map $(\tilde{F}_{D})_{q} =
(\tilde{F}_{D})_{|D^*_q}: D^*_q \to \bar{D}^*_{F(q)} \simeq
(D^*/G)_{F(q)}$ is a linear isometry, for all $q\in Q$. Therefore,
from Remark \ref{relIgLeMa} and Corollary \ref{relHam-Jaun-1}, it
follows
\begin{corollary}
Let ${\mathcal S}$ the set of the $1$-cocycles $\bar{\alpha}$ of
the skew-symmetric algebroid $\tau_{\bar{D}} = \tau_{D/G}: \bar{D}
= D/G \to \bar{Q} = Q/G$ which are solution of the nonholonomic
Hamilton-Jacobi equation
\[
d^{\bar{D}}((E_{\bar{L}})_{|\bar{D}} \circ \#_{{\mathcal
G}_{\bar{D}}} \circ \bar{\alpha}) = 0
\]
(respectively, which satisfy the strongest  condition
$(E_{\bar{L}})_{|\bar{D}} \circ \#_{{\mathcal G}_{\bar{D}}} \circ
\bar{\alpha} = \mbox{ constant }$). Then, there exists a
one-to-one correspondence between ${\mathcal S}$ and the following
sets:
\begin{enumerate}
\item
The set of the $G$-invariant $1$-cocycles $\alpha$ of the
skew-symmetric algebroid $\tau_{D}: D \to Q$ which are solutions
of the nonholonomic Hamilton-Jacobi equation
\[
d^D((E_{L})_{|D} \circ \#_{{\mathcal G}_{D}} \circ \alpha) = 0
\]
(respectively, which satisfy the strongest condition $(E_{L})_{|D}
\circ \#_{{\mathcal G}_{D}} \circ \alpha = \mbox{ constant }$).
\item
The set of the $G$-invariant $1$-forms $\gamma: Q \to
(D^{\perp})^0 \subseteq T^*Q$ on $Q$ which satisfy the following
conditions
\[
d\gamma \in {\mathcal I}(D^0) \; \mbox{ and } \; d(E_{L} \circ
\#_{g} \circ \gamma)(Q) \subseteq D^0
\]
(respectively, which satisfy the strongest conditions $d\gamma \in
{\mathcal I}(D^0)$ and $E_{L} \circ \#_{g} \circ \gamma = \mbox{
constant }$).
\end{enumerate}
\end{corollary}

\noindent {\it An explicit example:} \emph{The snakeboard}.

The snakeboard is a modified version of the traditional
skateboard, where the rider uses his own momentum, coupled with
the constraints, to generate forward motion. The configuration
manifold is $Q=SE(2)\times \T^2$ with coordinates $(x, y, \theta,
\psi, \phi)$ (see ~\cite{BuZe,MaKo}).

\begin{center}
\includegraphics{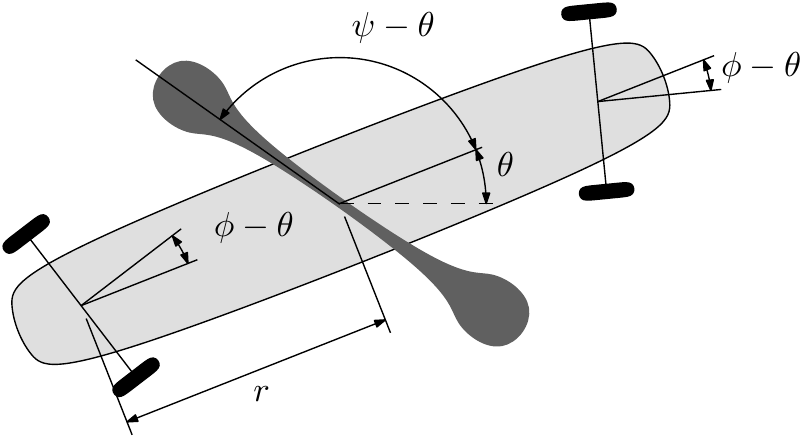}
\end{center}

The  system is described by a Lagrangian
\begin{eqnarray*}
 L(q, \dot{q})= \frac{1}{2}m (\dot{x}^2+\dot{y}^2) + \frac{1}{2}
 (J+2J_1)\dot{\theta}^2
 +\frac{1}{2}J_0(\dot{\theta}+\dot{\psi})^2
 +J_1\dot{\phi}^2
 \end{eqnarray*}
where $m$ is the total mass of the board, $J>0$ is the moment of
inertia of the board, $J_0>0$ is the moment of inertia of the
rotor of the snakeboard mounted on the body's center of mass and
$J_1>0$  is the moment of inertia of each wheel axles. The
distance between the center of the board and the wheels is denoted
by $r$. For simplicity, as in \cite{MaKo}, we assume that
$J+J_0+2J_1=mr^2$.

The inertia matrix representing the kinetic energy of the metric
$g$ on $Q$ defined by the snakeboard is
\[
g=mdx^2+mdy^2+mr^2d\theta^2+J_0 d\theta\otimes d\psi+J_0
d\psi\otimes d\theta + J_0d\psi^2+2J_1 d\phi^2.
\]
Since the wheels are not allowed to slide in the sideways
direction, we impose  the constraints
\begin{align*}
-\dot{x}\sin (\theta+\phi) +\dot{y}\cos (\theta+\phi)
-r\dot{\theta}\cos\phi&=0\\
-\dot{x}\sin (\theta-\phi) +\dot{y}\cos (\theta-\phi)
+r\dot{\theta}\cos\phi&=0.
\end{align*}
To avoid singularities of the distribution defined by the previous
constraints we will assume, in the sequel, that
$\phi\not=\pm\pi/2$.

Define the functions
\begin{eqnarray*}
a&=&-r(\cos\phi\cos(\theta-\phi)+\cos\phi\cos(\theta+\phi))=-2r\cos^2\phi\cos\theta\\
b&=&-r(\cos\phi\sin(\theta-\phi)+\cos\phi\sin(\theta+\phi))=-2r\cos^2\phi\sin\theta\\
c&=&\sin(2\phi).
\end{eqnarray*}
 The constraint  subbundle $\tau_D: D\longmapsto Q$ is
 \[
{D}=\text{span}\left\{ \frac{\partial}{\partial \psi},
\frac{\partial}{\partial \phi}, a\frac{\partial}{\partial
x}+b\frac{\partial}{\partial y}+c\frac{\partial}{\partial
\theta}\right\}.
\]

The Lagrangian function and the constraint subbundle
 are left-invariant under the $SE(2)$ action:
\[ \Phi_g(q)=(\alpha+x\cos\gamma-y\sin \gamma,
\beta+x\sin\gamma+y\cos \gamma, \gamma+\theta, \psi, \phi)
\]
where $g=(\alpha, \beta, \gamma)\in SE(2)$.
%(Observe that the Lagrangian is also $S^1$-invariant $\bar{\Phi}_g(q)=(x, y, \theta, \psi, \phi)$ but we do not take into account  this action  since %for control purposes is convenient  consider only the $SE(2)$-action).

We have a principal bundle structure $F: Q\longrightarrow \bar{Q}$
where $\bar{Q}=(SE(2)\times \T^2)/SE(2) \simeq \T^2$, being its
vertical bundle $ {\displaystyle VF=\hbox{span}\,
\left\{\frac{\partial}{\partial x}, \frac{\partial}{\partial y},
\frac{\partial}{\partial \theta}\right\}} $. We have that
\[{\displaystyle S=D\cap VF=\hbox{span}\,
\left\{Y_3=a\frac{\partial}{\partial x}+b\frac{\partial}{\partial
y}+c\frac{\partial}{\partial \theta}\right\}} \] and therefore,
\[
S^{\perp}\cap D=\hbox{span}\, \left\{Y_1=\frac{\partial}{\partial
\phi},  Y_2= \frac{\partial}{\partial
\psi}-\frac{J_0c}{k}Y_3\right\}=\hbox{span}\,
\left\{Y_1=\frac{\partial}{\partial \phi},
Y_2=\frac{\partial}{\partial \psi}-\frac{J_0}{2mr^2}(\tan \phi)
Y_3\right\}
\]
where $k=m(a^2+b^2+c^2r^2)=4mr^2(\cos^2\phi)$ (away form
$\phi=\pm\pi/2$).

Note that if $\{\xi_1, \xi_2, \xi_3\}$ is the canonical basis of
$\frak{se}(2)$
%\[
%[\xi_1, \xi_2]=0,\quad [\xi_1, \xi_3]=-\xi_2, \quad [\xi_2, \xi_3]=\xi_1,
%\]
then
\begin{eqnarray*}
Y_2&=&\frac{\partial}{\partial
\psi}-\frac{J_0\sin \phi}{mr^2}\left[-r(\cos\phi)\lvec{\xi_1}+(\sin \phi) \lvec{\xi_3}\right]\\
Y_3&=&-2r(\cos^2\phi)\lvec{\xi_1}+(\sin 2\phi) \lvec{\xi_3}\\
\end{eqnarray*}
where $\lvec{\xi_i}$ $(i=1, 2, 3$) is the left-invariant vector
field of $SE(2)$ such that $\lvec{\xi_i}(e)=\xi_i$, $e$ being the
identity element of $SE(2)$.

Next, we will denote by $\{X_1, X_2, X_3\}$ the $g$-orthonormal
basis  of $D$ given by
\begin{eqnarray*}
X_1&=&\frac{1}{\sqrt{2J_1}}\frac{\partial}{\partial \phi},\\
X_2&=&\frac{1}{\sqrt{f(\phi)}}\left(\frac{\partial}{\partial
\psi}-\frac{J_0\sin \phi}{mr^2}\left[-r(\cos\phi)\lvec{\xi_1}+(\sin \phi) \lvec{\xi_3}\right]\right)\\
X_3&=&\frac{1}{\sqrt{m}}\left[-(\cos\phi)\lvec{\xi_1}+\frac{1}{r}(\sin\phi) \lvec{\xi_3}\right], \\
\end{eqnarray*}
where $\displaystyle{f(\phi)=J_0-\frac{J_0^2\sin^2\phi}{mr^2}}.$
The vector fields $\{X_1, X_2\}$ describe changes in the internal
angles $\phi$ and $\psi$, while $X_3$ represents the instantaneous
rotation when the internal angles are fixed.

Consider now the corresponding Atiyah algebroid
\[
TQ/SE(2)\simeq (T\T^2\times T\,SE(2))/SE(2)\longrightarrow
\bar{Q}=\T^2.
\]
Using the left translations on $SE(2)$, we have that the tangent
bundle of $SE(2)$ may be identified with the product manifold
$SE(2)\times \frak{se}(2)$ and therefore the Atiyah algebroid is
identified with the  vector bundle
$\tilde{\tau}_{\T^2}=\tau_{\bar{A}}: \bar{A}=T\T^2\times
\frak{se}(2)\longrightarrow \T^2. $ The canonical basis of
$\tau_{\bar{A}}: T\T^2\times \frak{se}(2)\longrightarrow \T^2$ is
$ {\displaystyle \left\{ \frac{\partial}{\partial \psi},
\frac{\partial }{\partial \phi}, \xi_1, \xi_2, \xi_3\right\}.} $
The anchor map and the linear bracket of the Lie algebroid
$\tau_{\bar{A}}: T\T^2\times \frak{se}(2)\longrightarrow \T^2$ is
given by
\begin{eqnarray*}
&&\rho_{\bar{A}}(\frac{\partial}{\partial
\psi})=\frac{\partial}{\partial \psi}, \qquad
\rho_{\bar{A}}(\frac{\partial}{\partial
\phi})=\frac{\partial}{\partial \phi}, \qquad
\rho_{\bar{A}}(\xi_i)=0, \quad
i=1, 2, 3\\
&& \lcf\xi_1, \xi_3\rcf_{\bar{A}}=-\xi_2, \qquad \lcf\xi_2,
\xi_3\rcf_{\bar{A}}=\xi_1,
\end{eqnarray*}
being equal to zero the rest of the fundamental Lie brackets.

We select the orthonormal basis of sections,  $\{X'_1, X'_2, X'_3,
X'_4, X'_5\}$, where
\begin{eqnarray*}
{X}'_1&=&
\displaystyle\frac{1}{\sqrt{2J_1}}\frac{\partial}{\partial
\phi},\\[10pt]
{X}'_2&=&
\displaystyle\frac{1}{\sqrt{f(\phi)}}\left(\frac{\partial}{\partial
\psi}-\frac{J_0\sin \phi}{mr^2}\left[-r(\cos\phi)\xi_1+(\sin \phi) \xi_3\right]\right)\\[10pt]
X'_3&=& \displaystyle\frac{1}{\sqrt{m}}\left[-(\cos\phi)
\xi_1+\frac{1}{r}(\sin\phi) \xi_3\right],
\end{eqnarray*}
and  $\{X'_4, X'_5\}$ is an orthonormal  basis of sections of the
orthogonal complement to $\bar{D}$, $\bar{D}^{\perp}$, with
respect to the induced bundle metric ${\mathcal G}_{\bar{A}}$.

Taking the induced coordinates $(\psi, \phi, v^1, v^2, v^3, v^4,
v^5)$ on $T\T^2\times \frak{se}(2)$ by this basis of sections, we
deduce that the  space of orbits $\bar{D}$ of the action of
$SE(2)$ on $D$ has as local equations, $v^4=0$ and $v^5=0$,  being
a basis of sections of $\bar{D}$, $\{X'_1, X'_2, X'_3\}$.
Moreover, in these coordinates the reduced Lagrangian $\bar{L}:
T\T^2\times \frak{se}(2)\longrightarrow \R$ is
\[
\bar{L}=\frac{1}{2}\left(
(v^1)^2+(v^2)^2+(v^3)^2+(v^4)^2+(v^5)^2\right)\;.
\]
Now, we consider the reduced nonholonomic mechanical system
$(\bar{L}, \bar{D})$.

After, some straightforward computations  we deduce that
\[
\lcf X'_1, X'_2\rcf_{\bar{D}}= \displaystyle -\frac{J_0\cos
\phi}{r\sqrt{2J_1mf(\phi)}}X'_3, \qquad \lcf X'_1,
X'_3\rcf_{\bar{D}}= \displaystyle\frac{J_0\cos \phi}{r\sqrt{2J_1m
f(\phi)}}X'_2, \qquad \lcf X'_2, X'_3\rcf_{\bar{D}}=0\; .
\]

Therefore, the non-vanishing structure functions are:
\[
C_{12}^{3}=-C_{21}^3= \displaystyle-\frac{J_0\cos
\phi}{r\sqrt{2J_1mf(\phi)}},\qquad C_{13}^{2}=-C_{31}^2=
\displaystyle\frac{J_0\cos \phi}{r\sqrt{2J_1mf(\phi)}}.
\]
Moreover,\begin{eqnarray*} &\rho_{\bar{D}}(X'_1)=
\displaystyle\frac{1}{\sqrt{2J_1}}\frac{\partial}{\partial \phi},
\qquad \rho_{\bar{D}}(X'_2) =
\displaystyle\frac{1}{\sqrt{f(\phi)}}\frac{\partial}{\partial
\psi}, \qquad \rho_{\bar{D}}(X'_3)=0\; .&
\end{eqnarray*}
This  shows that $\rho_{\bar{D}}(\bar{D})=T_q\T^2$ and then the
skew-symmetric algebroid $\bar{D}\longrightarrow \T^2$ is
completely nonholonomic.

The  local expression of the vector field $\xi_{(\bar{L},
\bar{D})}$ is
\begin{eqnarray*}
\xi_{(\bar{L}, \bar{D})}&=&
\displaystyle\frac{v^1}{\sqrt{2J_1}}\frac{\partial}{\partial
\phi}+ \frac{v^2}{\sqrt{f(\phi)}}\frac{\partial}{\partial
\psi}\displaystyle-\frac{J_0\cos
\phi}{r\sqrt{2J_1mf(\phi)}}v^1v^3\frac{\partial}{\partial v^2}
+\displaystyle\frac{J_0\cos
\phi}{r\sqrt{2J_1mf(\phi)}}v^1v^2\frac{\partial}{\partial v^3}
\end{eqnarray*}
Let $\{(X')^1, (X')^2, (X')^3\}$ be the dual basis of $\bar{D}^*$.
It induces a local coordinate system: $(\phi, \psi, p_1, p_2,
p_3)$ on $\bar{D}^*$ and, therefore,   the non-vanishing terms of
the nonholonomic bracket are:
\begin{eqnarray*}
&\{\phi, p_1\}_{\bar{D}^*}=  \displaystyle\frac{1}{\sqrt{2J_1}},
\qquad \{\psi, p_2\}_{\bar{D}^*}=\frac{1}{\sqrt{f(\phi)}},&\\
&\{p_1, p_2\}_{\bar{D}^*}=  \displaystyle\frac{J_0\cos \phi}{r\sqrt{2J_1mf(\phi)}}p_3,&\\
&\{p_1, p_3\}_{\bar{D}^*}= \displaystyle -\frac{J_0\cos
\phi}{r\sqrt{2J_1m f(\phi)}}p_2.&
\end{eqnarray*}

Now, we study the Hamilton-Jacobi equations for the snakeboard
system. A section $\alpha: \T^2\longrightarrow \bar{D}^*$,
$\alpha= \alpha_1(\phi, \psi)(X')^1+\alpha_2(\phi, \psi)(X')^2+
\alpha_3(\phi, \psi)(X')^3$, is a 1-cocycle
($d^{\bar{D}}\alpha=0$) if and only if:
  \begin{eqnarray}
 0&=& \frac{1}{\sqrt{2J_1}}\frac{\partial \alpha_2}{\partial \phi}-\frac{1}{\sqrt{f(\phi)}}\frac{\partial \alpha_1}{\partial
\psi}+\displaystyle\frac{J_0\cos
\phi}{r\sqrt{2J_1mf(\phi)}}\alpha_3\label{equation-ref1}
  \\
 0&=& \frac{1}{\sqrt{2J_1}}\frac{\partial \alpha_3}{\partial \phi}\displaystyle -\frac{J_0\cos
\phi}{r\sqrt{2J_1m f(\phi)}}\alpha_2\label{equation-ref2}\\
 0&=&\frac{1}{\sqrt{f(\phi)}}\frac{\partial\alpha_3}{\partial
\psi}\label{equation-ref3}\; .
\end{eqnarray}

Finally, since the skew-symmetric algebroid is completely
nonholonomic, the Hamilton-Jacobi equation is rewritten as
\begin{equation}\label{equation-ref4}
(\alpha_1(\phi, \psi))^2 +(\alpha_2(\phi, \psi))^2 +
(\alpha_3(\phi, \psi))^2=\hbox{constant}
\end{equation}

Now, we will use this  equation for studying   explicit solutions
for the snakeboard, showing the availability of our methods for
obtaining new insights in nonholonomic dynamics.

{}From Equation (\ref{equation-ref3}) we obtain that
$\alpha_3=\alpha_3(\phi)$. Then it is clear that also
$\alpha_2=\alpha_2(\phi)$. Assume that $\alpha_1=\hbox{constant}$.
Therefore, Equations (\ref{equation-ref1}) and
(\ref{equation-ref2}) are now, in this case, a system of ordinary
differential equations:
\begin{eqnarray}
 0&=& \frac{d \alpha_2}{d \phi}+\displaystyle\frac{J_0\cos \phi}{r\sqrt{mf(\phi)}}\alpha_3
 \label{equation-ref1-1}
  \\
 0&=& \frac{d \alpha_3}{d \phi}\displaystyle -\frac{J_0\cos
\phi}{r\sqrt{m f(\phi)}}\alpha_2\; .\label{equation-ref2-2}
\end{eqnarray}
Moreover, observe that all the solutions of these equations
automatically satisfy Equation (\ref{equation-ref4}) since
\[
\alpha_2\frac{d \alpha_2}{d \phi} +
 \alpha_3\frac{d \alpha_3}{d \phi} = 0\, \hbox{  and  } \alpha_1=\hbox{constant}
\]
Solving explicitly the system of equations (\ref{equation-ref1-1})
and (\ref{equation-ref2-2}) we obtain that
\begin{eqnarray*}
\alpha_2(\phi)&=&C_1\sqrt{ f(\phi)}+\frac{J_0C_2}{r\sqrt{m}}\sin\phi\\
\alpha_3(\phi)&=&\frac{J_0C_1}{r\sqrt{m}}\sin\phi-C_2\sqrt{
f(\phi)}
\end{eqnarray*}
with $C_1, C_2$ arbitrary constants. Therefore,
\[
\alpha(\phi, \psi)=(\phi, \psi; \sqrt{2J_1}C_0, C_1\sqrt{ f(\phi)}
+\frac{J_0C_2}{r\sqrt{m}}\sin\phi,
\frac{J_0C_1}{r\sqrt{m}}\sin\phi-C_2\sqrt{ f(\phi)})
\]
is an 1-cocycle of the skew-symmetric algebroid $\bar{D}\to \T^2$,
for all $(C_0,C_1, C_2)\in \R^3$ and moreover it satisfies
Equations (\ref{equation-ref4}). Hence, we can use Corollary
\ref{nonhoHam-Jaeq} to obtain solutions of the reduced snakeboard
problem. First, we calculate the integral curves of the vector
field $\xi_{(\bar{L}, \bar{D})\alpha}$:
\begin{eqnarray*}
\dot{\phi}(t)&=&C_0\\
\dot{\psi}(t)&=&\frac{1}{\sqrt{f(\phi(t))}}\left(C_1\sqrt{ f(\phi(t))}+\frac{J_0C_2}{r\sqrt{m}}\sin\phi(t)\right)\\
&=&C_1 + \frac{J_0C_2}{r\sqrt{m f(\phi(t))}}\sin\phi(t)
\end{eqnarray*}
whose solutions are:
\begin{eqnarray*}
\phi(t)&=&C_0t+C_3\\
\psi(t)&=&C_1t  -\frac{C_2}{C_0}\log\left[
\sqrt{2}\left(\sqrt{J_0}\cos (C_0 t+C_3)
+\sqrt{mr^2-J_0\sin^2 (C_0t+C_3)}\right)\right]+C_4\quad \hbox{ (if $C_0\not=0$)}\\
\psi(t)&=&C_1t + \frac{\sqrt{J_0}C_2 t
\sin(C_3)}{\sqrt{mr^2-J_0\sin^2(C_3)}}+C_4\quad \hbox{ (if
$C_0=0$)}
\end{eqnarray*}
for all constants $C_i\in \R$, $1\leq i\leq 5$. Now, by a direct
application of the nonholonomic equation we obtain that
\begin{eqnarray*}
%\phi(t)&=&C_0t+C_4\\
%\psi(t)&=&C_1t + \frac{\sqrt{J_0}C_2 t \sin(C_0t+C_4)}{\sqrt{mr^2-J_0\sin^2(C_0t+C_4)}}+C_5\\
v^1(t)&=&\sqrt{2J_1}C_0\\
v^2(t)&=&C_1\sqrt{ f(C_0t+C_3)}+\frac{J_0C_2}{r\sqrt{m}}\sin(C_0t+C_3)\\
v^3(t)&=&\frac{J_0C_1}{r\sqrt{m}}\sin(C_0t+C_3)-C_2\sqrt{
f(C_0t+C_3)}
\end{eqnarray*}
are solutions of the reduced nonholonomic problem.

%A solution is, for instance, $\alpha_1=k$, $\alpha_2=0$ and $\alpha_3=0$, with $k\in \R$.%
%, which corresponds to a rotation of the snakeboard around its instantaneous center of rotation when there are no motion of the rider.

\section{Conclusions and Future Work}
In this paper we have elucidated the geometrical framework for the
Hamilton-Jacobi equation. Our formalism is valid for nonholonomic
mechanical systems. The basic geometric ingredients are a vector
bundle, a linear almost Poisson bracket and a Hamiltonian function
both on the dual bundle.  We also have discussed the behavior of
the theory under Hamiltonian morphisms and its applicability to
reduction theory. Some examples are studied in detail and, as a
consequence, it is shown the utility of our framework to integrate
the dynamical equations. However, in this direction more work
must be done.

In particular, as a future research, we will study new particular
examples, testing candidates for solutions of the nonholonomic
Hamilton-Jacobi equation of the form $\alpha = d^D f$, for some
$f\in C^{\infty}(Q)$ (if there exists) and  moreover we will study
the complete solutions for the Hamilton-Jacobi equation using the
groupoid theory. In this line, we will study the construction of
numerical integrators via Hamilton-Jacobi theory \cite{Hair}.
 We will
also discuss the extension of our formalism to time-dependent
Lagrangian systems subjected to affine constraints in the
velocities. It would be interesting to describe the
Hamilton-Jacobi theory for variational constrained problems,
giving a geometric interpretation of the Hamilton-Jacobi-Bellman
equation for optimal control systems. Finally, extensions to
classical field theories in the present context could be
developed.

\appendix

\section*{Appendix}

\newcounter{apendice}
\setcounter{apendice}{0}
\renewcommand{\thesection}{A}
\setcounter{equation}{0} \setcounter{apendice}{1}

Let $\{\cdot, \cdot\}_{D^*}$ be a linear almost Poisson structure
on a vector bundle $\tau_{D}: D \to Q$, $(\lcf \cdot, \cdot
\rcf_{D}, \rho_{D})$ be the corresponding skew-symmetric Lie
algebroid structure on $D$ and $\alpha: Q \to D^*$ be a section of
$\tau_{D^*}: D^* \to Q$. If $q \in Q$ then we may choose local
coordinates $(q^U) = (q^i, q^a)$ on an open subset $U$ of $Q$,
$q\in U$, and a basis of sections $\{X_{A}\} = \{X_{i},
X_{\gamma}\}$ of the vector bundle $\tau_{D}^{-1}(U) \to U$ such
that
\begin{equation}\label{anclaenq}
\rho_{D}(X_{i})(q) = \displaystyle \frac{\partial}{\partial
q^i}_{|q}, \; \; \; \rho_{D}(X_{\gamma})(q) = 0.
\end{equation}
Suppose that
\begin{equation}\label{ecsdeestructura}
\rho_{D}(X_{A}) = \rho^{U}_{A}\displaystyle
\frac{\partial}{\partial q^U}, \; \; \; \lcf X_{A}, X_{B}\rcf_{D}
= C_{AB}^{C}X_{C}
\end{equation}
and that the local expression of $\alpha$ in $U$ is
\begin{equation}\label{exprelocalalpha}
\alpha(q^U) = (q^U, \alpha_{A}(q^U)).
\end{equation}
Denote by $\Lambda_{D^*}$ the linear almost Poisson $2$-vector on
$D^*$ and by $(q^U, p_{A}) = (q^i, q^a, p_{i}, p_{\gamma})$ the
corresponding local coordinates on $D^*$. Then, from
(\ref{LambdaDstar}), it follows that
\begin{equation}\label{2-vectorenelpunto}
\Lambda_{D^*}(\alpha(q)) = \displaystyle \frac{\partial}{\partial
q^i}_{|\alpha(q)} \wedge \frac{\partial}{\partial
p_{i}}_{|\alpha(q)} - \frac{1}{2} C_{AB}^{C}(q)\alpha_{C}(q)
\frac{\partial}{\partial p_{A}}_{|\alpha(q)} \wedge
\frac{\partial}{\partial p_{B}}_{|\alpha(q)}.
\end{equation}
Moreover, using (\ref{anclaenq}), (\ref{ecsdeestructura}) and
(\ref{exprelocalalpha}), we obtain that
\begin{equation}\label{difalpha}
\begin{array}{rcl}
(d^D\alpha)(q)(X_{i}(q), X_{j}(q)) &=& \displaystyle
\frac{\partial \alpha_{j}}{\partial q^i}_{|q} - \frac{\partial
\alpha_{i}}{\partial q^j}_{|q} - C_{ij}^{A}(q)\alpha_{A}(q), \\
[5pt] (d^D\alpha)(q)(X_{i}(q), X_{\gamma}(q)) &=& \displaystyle
\frac{\partial \alpha_{\gamma}}{\partial q^i}_{|q} -
C_{i\gamma}^{A}(q)\alpha_{A}(q),\\ [10pt]
(d^D\alpha)(q)(X_{\gamma}(q), X_{\nu}(q)) & =& -
C_{\gamma\nu}^{A}(q)\alpha_{A}(q).
\end{array}
\end{equation}
On the other hand, let ${\mathcal L}_{\alpha, D}(q)$ be the
subspace of $T_{\alpha(q)}D^*$ defined by (\ref{DefLalD}). Then,
from (\ref{anclaenq}) and (\ref{exprelocalalpha}), we deduce that
\begin{equation}\label{basedeLalD}
{\mathcal L}_{\alpha, D}(q) = \langle\{\displaystyle
\frac{\partial}{\partial q^i}_{|\alpha(q)} + \frac{\partial
\alpha_{A}}{\partial q^i}_{|q}\frac{\partial}{\partial
p_{A}}_{|\alpha(q)}\}\rangle
\end{equation}
which implies that
\begin{equation}\label{baseanulador}
({\mathcal L}_{\alpha, D}(q))^{0} = \langle\{dq^a(\alpha(q)),
dp_{j}(\alpha(q)) - \displaystyle \frac{\partial
\alpha_{j}}{\partial q^i}_{|q} dq^i(\alpha(q)),
dp_{\gamma}(\alpha(q)) - \displaystyle \frac{\partial
\alpha_{\gamma}}{\partial q^i}_{|q} dq^i(\alpha(q))\}\rangle.
\end{equation}
In addition, using (\ref{2-vectorenelpunto}), one may prove that
\begin{equation}\label{basedelortg}
\begin{array}{rcl}
\#_{\Lambda_{D^*}}(dq^a(\alpha(q))) & = & 0, \\ [5pt]
\#_{\Lambda_{D^*}}(dp_{j}(\alpha(q)) - \displaystyle
\frac{\partial \alpha_{j}}{\partial q^i}_{|q} dq^i(\alpha(q))) &=&
-\displaystyle \frac{\partial}{\partial q^j}_{|\alpha(q)} -
(\frac{\partial \alpha_{j}}{\partial q^i}_{|q} -
C_{ij}^{C}(q)\alpha_{C}(q))\frac{\partial}{\partial
p_{i}}_{|\alpha(q)}\\ [10pt] && \displaystyle -
C^C_{j\gamma}(q)\alpha_{C}(q) \frac{\partial}{\partial
p_{\gamma}}_{|\alpha(q)}, \\ [10pt]
\#_{\Lambda_{D^*}}(dp_{\gamma}(\alpha(q)) - \displaystyle
\frac{\partial \alpha_{\gamma}}{\partial q^i}_{|q}
dq^i(\alpha(q))) &=& -\displaystyle (\frac{\partial
\alpha_{\gamma}}{\partial q^i}_{|q} -
C_{i\gamma}^{C}(q)\alpha_{C}(q))\frac{\partial}{\partial
p_{i}}_{|\alpha(q)}\\[10pt]&& \displaystyle - C^C_{\gamma\nu}(q)\alpha_{C}(q)
\frac{\partial}{\partial p_{\nu}}_{|\alpha(q)}.
\end{array}
\end{equation}

\noindent {\it Proof of Proposition \ref{primerresultado}.} From
(\ref{difalpha}), (\ref{basedeLalD}), (\ref{baseanulador}) and
(\ref{basedelortg}), we deduce the result. \hfill$\Box$

\noindent {\it Proof of Proposition \ref{segundoresultado}.}
Suppose that
\[
\beta_{\alpha(q)} = \lambda_{U}dq^{U}(\alpha(q)) +
\mu^{A}dp_{A}(\alpha(q)) \in T_{\alpha(q)}^*D^*.
\]
Then, using (\ref{2-vectorenelpunto}) and (\ref{difalpha}), it
follows that
\[
\beta_{\alpha(q)} \in Ker \#_{\Lambda_{D^*}}(\alpha(q))
\Longleftrightarrow \mu^i = 0, \; \; \lambda_{i} = -\displaystyle
\frac{\partial \alpha_{\gamma}}{\partial q^i}_{|q} \mu^{\gamma},
\; \; \mbox{ for all } i.
\]
Thus, from (\ref{baseanulador}), we conclude that
\[
\beta_{\alpha(q)} \in ({\mathcal L}_{\alpha, D}(q))^0.
\]
\hfill$\Box$

\end{document}